\documentclass[11pt, letterpaper]{article}
\usepackage[utf8]{inputenc}
\usepackage{fullpage,times}

\usepackage[procnumbered,ruled,vlined,linesnumbered]{algorithm2e}

\DontPrintSemicolon
\SetKw{KwAnd}{and}
\SetProcNameSty{textsc}
\SetFuncSty{textsc}

\usepackage{amsmath,amsfonts,amsthm,amssymb,multirow}
\usepackage{times}
\usepackage{graphicx}
\usepackage{floatpag}
\usepackage{bbold}
\usepackage{enumitem}
\usepackage{subcaption} 
\usepackage{calc}
\usepackage{tikz}
\usetikzlibrary{decorations.markings}
\tikzstyle{vertex}=[circle, draw, inner sep=0pt, minimum size=4pt, fill = black]

\usepackage{graphicx}
\usepackage{tabularx}
\makeatletter
\newcommand{\multiline}[1]{%
  \begin{tabularx}{\dimexpr\linewidth-\ALG@thistlm}[t]{@{}X@{}}
    #1
  \end{tabularx}
}
\makeatother
\usepackage{mathtools}

\makeatletter
\def\BState{\State\hskip-\ALG@thistlm}
\makeatother

\newcommand{\floor}[1]{\lfloor #1 \rfloor}
\newcommand{\rd}[2]{d(#1\leftrightarrows #2)}
\newcommand{\rdh}[3]{d_{#1}(#2\leftrightarrows #3)}

\usepackage[compact]{titlesec}
\titlespacing{\section}{0pt}{3ex}{2ex}
\titlespacing{\subsection}{0pt}{2ex}{1ex}
\titlespacing{\subsubsection}{0pt}{0.5ex}{0ex}

\newtheorem{theorem}{Theorem}[section]

\newenvironment{proofof}[1]{{\bf Proof of #1.  }}{\hfill$\Box$}

\newtheorem{lemma}{Lemma}[section]
\newtheorem{claim}{Claim}
\newtheorem{hypothesis}{Hypothesis}

\makeatletter
\let\c@fconjecture\c@conjecture
\makeatother

\makeatletter
\let\c@fconj\c@conj
\makeatother

\def \qed {\hfill $\Box$}

\def \eps {\varepsilon}

\newcommand{\ignore}[1]{}

\def\tO{\tilde{O}}

\title{Conditionally optimal approximation algorithms \\for the girth of a directed graph}
\author{Mina Dalirrooyfard\\{MIT, minad@mit.edu} \and Virginia Vassilevska Williams\\ {MIT, virgi@mit.edu} }
\date{}

\begin{document}

\maketitle
\begin{abstract}
The girth is one of the most basic graph parameters, and its computation has been studied for many decades.
Under widely believed fine-grained assumptions, computing the girth exactly is known to require $mn^{1-o(1)}$ time, both in sparse and dense $m$-edge, $n$-node graphs, motivating the search for fast approximations. Fast good quality approximation algorithms for undirected graphs have been known for decades. For the girth in directed graphs, until recently the only constant factor approximation algorithms ran in $O(n^\omega)$ time, where $\omega<2.373$ is the matrix multiplication exponent. These algorithms have two drawbacks: (1) they only offer an improvement over the $mn$ running time for dense graphs, and (2) the current fast matrix multiplication methods are impractical.
The first constant factor approximation algorithm that runs in $O(mn^{1-\eps})$ time for $\eps>0$ and all sparsities $m$ was only recently obtained by Chechik et al. [STOC 2020]; it is also combinatorial.

It is known that a better than $2$-approximation algorithm for the girth in dense directed unweighted graphs needs $n^{3-o(1)}$ time unless one uses fast matrix multiplication. Meanwhile, the best known approximation factor for a combinatorial algorithm running in $O(mn^{1-\eps})$ time (by Chechik et al.) is $3$. Is the true answer $2$ or $3$?





The main result of this paper is a (conditionally) tight approximation algorithm for directed graphs.
First, we show that under a popular hardness assumption, any algorithm, even one that exploits fast matrix multiplication, would need to take at least $mn^{1-o(1)}$ time for some sparsity $m$ if it achieves a $(2-\eps)$-approximation for any $\eps>0$. Second we give
 a $2$-approximation algorithm for the girth of unweighted graphs running in $\tilde{O}(mn^{3/4})$ time, and a $(2+\eps)$-approximation algorithm (for any $\eps>0$) that works in weighted graphs and runs in $\tilde{O}(m\sqrt n)$ time. Our algorithms are combinatorial. 
 

We also obtain a $(4+\eps)$-approximation of the girth running in $\tilde{O}(mn^{\sqrt{2}-1})$ time, improving upon the previous best $\tilde{O}(m\sqrt n)$ running time by Chechik et al.
Finally, we consider the computation of roundtrip spanners. We obtain a $(5+\eps)$-approximate roundtrip spanner on $\tilde{O}(n^{1.5}/\eps^2)$ edges in $\tilde{O}(m\sqrt n/\eps^2)$ time. This improves upon the previous approximation factor $(8+\eps)$ of Chechik et al. for the same running time.
\end{abstract}
\thispagestyle{empty}
\newpage
\setcounter{page}{1}

\section{Introduction.}
One of the most basic and well-studied graph parameters is the {\em girth}, i.e. the length of the shortest cycle in the graph. Computing the girth in an $m$-edge, $n$-node graph can be done by computing all pairwise distances, that is, solving the All-Pairs Shortest Paths (APSP) problem. This gives an $\tilde{O}(mn)$ time algorithm for the general version of the girth problem: directed or undirected integer weighted graphs and no negative weight cycles\footnote{If the weights are nonnegative, running Dijkstra's algorithm suffices. If there are no negative weight cycles, one can use Johnson's trick to make the weights nonnegative at the cost of a single SSSP computation which can be achieved for instance in $\tilde{O}(m\sqrt n\log M)$ time if $M$ is the largest edge weight magnitude via Goldberg's algorithm \cite{goldberg}, so as long as the weights have at most $\tilde{O}(\sqrt n)$ bits, the total time is $\tilde{O}(mn)$.}. 

The $\tilde{O}(mn)$ running time for the exact computation of the girth is known to be tight, up to $n^{o(1)}$ factors, both for sparse and dense weighted graphs, under popular hardness hypotheses from fine-grained complexity \cite{focsy,lincolnsoda}. In unweighted graphs or graphs with integer weights of magnitude at most $M$, one can compute the girth in $\tilde{O}(Mn^\omega)$ time \cite{seidel,Clique1,RodittyVW11,cyganbaur} where $\omega<2.373$ is the exponent of $n\times n$ matrix multiplication \cite{vstoc12,legallmult}.
This improves upon $mn$ only for somewhat dense graphs with small weights, and moreover is not considered very practical due to the large overhead of fast matrix multiplication techniques. 

Due to the subcubic equivalences of \cite{focsy}, however, it is known that even in unweighted dense graphs, any algorithm that computes the girth in $O(n^{3-\eps})$ time needs to use fast matrix multiplication techniques, unless one can obtain a subcubic time combinatorial Boolean Matrix Multiplication (BMM) algorithm. Thus, under popular fine-grained complexity assumptions, if one wants to have a fast combinatorial algorithm, or an algorithm that is faster than $mn$ for sparser graphs, one needs to resort to {\em approximation}.

Fast approximation algorithms for the girth in undirected graphs have been known since the 1970s, starting with the work of Itai and Rodeh \cite{Clique1}. The current strongest result shows a $2$-approximation in $\tilde{O}(n^{5/3})$ time \cite{RodittyW12}; note that if the graph is dense enough this algorithm is sublinear in the input. Such good approximation algorithms are possible for undirected graphs because of known strong structural properties. For instance, as shown by Bondy and Simonovits \cite{BoSi74}, for any integer $k\geq 2$, if a graph has at least $100kn^{1+1/k}$ edges, then it must contain a $2k$ cycle, and this gives an immediate upper bound on the girth. 
There are no such structural results for directed graphs, making the directed girth approximation problem quite challenging.

Zwick~\cite{zwickbridge} showed that if the maximum weight of an edge is $M$, one can obtain in $\tilde{O}(n^\omega\log(M/\eps)/\eps)$ time a $(1+\eps)$-approximation for APSP, and this implies the same for the girth of directed graphs. As before, however, this algorithm does not run fast in sparse graphs, and can be considered impractical.

The first nontrivial approximation algorithms (both for sparse graphs and combinatorial) for the girth of directed graphs were achieved by Pachocki et al. \cite{PachockiRSTW18}. The current best result by Chechik et al. \cite{improvedgirtharx,improvedgirthSTOC} achieves for every integer $k\geq 1$, a randomized $O(k\log k)$-approximation algorithm running in time $\tilde{O}(m^{1+1/k})$. The best approximation factor that Chechik et al. obtain in $O(mn^{1-\eps})$ time for $\eps>0$ is $3$, in $\tilde{O}(m\sqrt n)$ time.

What should be the best approximation factor attainable in $O(mn^{1-\eps})$ time for $\eps>0$? 
It is not hard to show (see e.g. \cite{vthesis}, the construction in Thm 4.1.3) that graph triangle detection can be reduced to triangle detection in a directed graph whose cycle lengths are all divisible by $3$. This, coupled with the combinatorial subcubic equivalence between triangle detection and BMM \cite{focsy} implies that any $O(n^{3-\eps})$ time algorithm for $\eps>0$ that achieves a $(2-\delta)$-approximation for the girth implies an $O(n^{3-\eps/3})$ time algorithm for BMM, and hence fast matrix multiplication techniques are likely necessary for faster $(2-\eps)$-approximation of the directed girth.

\subsection{Our results}
We first give a simple extension to the above hardness argument for $(2-\eps)$-approximation, giving a conditional lower bound on the running time of $(2-\eps)$-girth approximation algorithms under the so called $k$-Cycle hardness hypothesis \cite{ancona2019,dynssspMaxNicole20,lincolnsoda}. 

The $k$-Cycle hypothesis states that for every $\eps>0$, there is a $k$ such that $k$-cycle in $m$-edge directed unweighted graphs cannot be solved in $O(m^{2-\eps})$ time (on a $O(\log n)$ bit word-RAM). 

The hypothesis is consistent with all known algorithms for detecting $k$-cycles in directed graphs, as these run at best in time $m^{2-c/k}$ for various small constants $c$ \cite{YuZw04,AlYuZw97,lincolnsoda,patternscycles19}, even using powerful tools such as matrix multiplication.  Moreover, as shown by Lincoln et al. \cite{lincolnsoda} any $O(mn^{1-\eps})$ time algorithm (for $\eps>0$) that, for odd $k$, can detect $k$-cycles in $n$-node $m$-edge directed graphs with
$m=\Theta(n^{1+2/(k-1)})$, would imply an $O(n^{k-\delta})$ time algorithm for $k$-clique detection for $\delta>0$. If the cycle algorithm is ``combinatorial'', then the clique algorithm would be ``combinatorial'' as well, and since all known $O(n^{k-\delta})$ time $k$-clique algorithms use fast matrix multiplication, such a result for $k$-cycle would be substantial.

In Section \ref{sec:hardness}, with a very simple reduction we show:

\begin{theorem}
\label{thm:lowerbound}
Suppose that for some constants $\eps>0$ and $\delta>0$, there is an $O(m^{2-\eps})$ time algorithm that can compute a $(2-\delta)$-approximation of the girth in an $m$-edge directed graph. Then for every constant $k$, one can detect whether an $m$-edge directed graph contains a $k$-cycle, in $O(m^{2-\eps})$ time, and hence the $k$-Cycle Hypothesis is false.
\end{theorem}

Thus, barring breakthroughs in Cycle and Clique detection algorithms, we know that the best we can hope for using an $O(mn^{1-\eps})$ time algorithm for the girth of directed graphs is a $2$-approximation. The proof of Theorem \ref{thm:lowerbound} is presented in section \ref{sec:hardness}.

The main result of this paper is the first ever $O(mn^{1-\eps})$ time for $\eps>0$ $2$-approximation algorithm for the girth in directed graphs. This result is conditionally tight via the above discussion.

\begin{theorem}
\label{thm:2approxunw}
There is an $\tilde{O}(mn^{3/4})$ time randomized algorithm that $2$-approximates the girth in directed unweighted graphs whp. For every $\eps>0$, there is a $(2+\eps)$-approximation algorithm for the girth in directed graphs with integer edge weights that runs in $\tilde{O}(m\sqrt{n}/\eps)$ time. The algorithms are randomized and are correct whp.
\end{theorem}

If one wanted to obtain a $(4+\eps)$-approximation to the girth via Chechik et al.'s $O(k\log k)$ approximation algorithms, the best running time one would be able to achieve is $\tilde{O}(m\sqrt n)$. Here we show how to get an improved running time for a $(4+\eps)$ approximation.

\begin{theorem}
\label{thm:4approx}
For every $\eps>0$, there is a $(4+\eps)$-approximation algorithm for the girth in directed graphs with integer edge weights that runs in $\tilde{O}(mn^{\sqrt{2}-1} /\eps)$ time. The algorithm is randomized and correct whp.
\end{theorem}

In fact, we obtain a generalization of the above algorithms that improves upon the algorithms of Chechik et al. for all constants $k$. 

\begin{theorem}\label{thm:genk}
For every $\eps>0$ and integer $k\geq 1$, there is a $(2k+\eps)$-approximation algorithm for the girth in directed graphs with integer edge weights that runs in $\tilde{O}(mn^{\alpha_k} /\eps)$ time, where $\alpha_k>0$ is the solution to $\alpha_k(1+\alpha_k)^{k-1}=1-\alpha_k$. The algorithms are randomized and correct whp.
\end{theorem}

For example, let's consider $\alpha_1$ in the above theorem. 
It is the solution to $\alpha_1=1-\alpha_1$, giving $\alpha_1=1/2$ and recovering the result of Theorem~\ref{thm:2approxunw} for weighted graphs. On the other hand, $\alpha_2$ is the solution to $\alpha_2(1+\alpha_2)=1-\alpha_2$, which gives $\alpha_2 = \sqrt{2}-1$ and recovering Theorem~\ref{thm:4approx}. Finally, say we wanted to get a $6+\eps$ approximation, then we need $\alpha_3$, which is the solution to $\alpha_3(1+\alpha_3)^2=1-\alpha_3$, giving $\alpha_3\leq 0.354$, and thus there's an $\tilde{O}(mn^{0.354} /\eps)$ time $(6+\eps)$-approximation algorithm. Note that there is only one positive solution to the equation defining $\alpha_k$ in Theorem \ref{thm:genk}.

As $k$ grows, $\alpha_k$ grows as $\Theta(\log k / k)$, and so the algorithm from Theorem~\ref{thm:genk} has similar asymptotic guarantees as the algorithm of Chechik et al. as it achieves an $O(\ell\log \ell)$ approximation in $\tilde{O}(mn^{1/\ell})$ time. The main improvements lie in the improved running time for small constant approximation factors.

%

Our approximation algorithms on weighted graphs can be found in section \ref{sec:weighted}. If we are aiming for an algorithm running in $T(n,m)$ time, we first suppose that the maximum edge weight of the graph is $M$ and we obtain an algorithm in $T(n,m)\log{M}$ time. We then show how to remove the $\log{M}$ factor at the end of section \ref{sec:weighted}. 

\noindent {\bf Roundtrip Spanners.}
Both papers that achieved nontrivial combinatorial approximation algorithms for the directed girth were also powerful enough to compute sparse approximate roundtrip spanners. 

A $c$-approximate roundtrip spanner of a directed graph $G=(V,E)$ is a subgraph $H=(V,E')$ of $G$ such that for every $u,v\in V$, $d_H(u,v)+d_H(v,u)\leq c\cdot (d_G(u,v)+d_G(v,u))$. 
Similar to what is known for spanners in undirected graphs, it is known \cite{duanroundtrip} that for every integer $k\geq 2$ and every $n$, every $n$-node graph contains a $(2k-1+o(1))$-approximate roundtrip spanner on $O(kn^{1+1/k}\log n)$ edges; the $o(1)$ error can be removed if the edge weights are at most polynomial in $n$ and the result then is optimal, up to log factors under the Erd\"os girth conjecture.

The best algorithms to date for computing sparse roundtrip spanners, similarly to the girth, achieve an $O(k\log k)$ approximation in $\tilde{O}(m^{1+1/k})$ time \cite{improvedgirthSTOC}. The best constant factor approximation achieved for roundtrip spanners in $O(mn^{1-\eps})$ time for $\eps>0$ is again achieved by Chechik et al.: a $(8+\eps)$ approximate  $O(n^{1.5})$-edge (in expectation) roundtrip spanner can be computed in $\tilde{O}(m\sqrt n)$ expected time. We improve this latter result:

\begin{theorem}
\label{thn:weighted_spanner}
There is an $\tilde{O}(m\sqrt n \log^2(M)/\eps^2)$ time randomized algorithm that computes a $(5+\eps)$-approximate roundtrip spanner on $\tilde{O}(n^{1.5}\log^2(M) /\eps^2)$ edges whp, for any $n$-node $m$-edge directed graph with edge weights in $\{1,\ldots,M\}$.
\end{theorem}

\section{Preliminary Lemmas}
We begin with some preliminary lemmas. The first two will allow us to decrease all degrees to roughly $m/n$, while keeping the number of vertices and edges roughly the same. The last lemma, implicit in \cite{improvedgirtharx}, is a crucial ingredient in our algorithms.

The following lemma was proven by Chechik et al. \cite{improvedgirtharx}:

\begin{lemma}\label{lemma:zerowt}
Given a directed graph $G=(V,E)$ with $|V|=n,|E|=m$, we can in $O(m+n)$ time construct a graph $G'=(V',E')$ with $V\subseteq V'$, so that $|V'|\leq O(n)$, $|E'|\leq O(m+n)$, for every $v\in V'$, $deg(v)\leq \lceil m/n\rceil$,
and so that for every $u,v\in V$, $d_{G'}(u,v)=d_G(u,v)$, and so that any path $p$ between some nodes $u\in V$ and $v\in V$ in $G'$ (possibly $u=v$) is in one-to-one correspondence with a path in $G$ of the same length.
\end{lemma}

The proof of the above lemma introduces edges of weight $0$, even if the graph was originally unweighted. 
In the lemma below which is proved in the appendix, we show how for an unweighted graph we can achieve essentially the same goal, but without adding weighted edges. This turns out to be useful for our unweighted girth approximation. 

\begin{lemma}
Given a directed unweighted graph $G=(V,E)$ and $|V|=n,|E|=m$, we can in $\tilde{O}(m+n)$ time construct an unweighted graph $G'=(V',E')$ with $V\subseteq V'$, so that $|V'|\leq O(n\log n)$, $|E'|\leq O(m+n\log n)$, for every $v\in V'$ $\text{out-deg}(v)\leq \lceil m/n\rceil$, and so that there is an integer $t$ such that for every $u,v\in V$, $d_{G'}(u,v)=t\cdot d_G(u,v)$, and so that any path $p$ between some nodes $u\in V$ and $v\in V$ in $G'$ (possibly $u=v$) is in one-to-one correspondence with a path in $G$ of length $1/t$ of the length of $p$.\label{lemma:unwtd}
\end{lemma}

In particular, the lemma will imply that the girth of $G'$ is exactly $t$ times the girth of $G$, and that given a $c$-roundtrip spanner of $G'$, one can in 
$\tilde{O}(m+n)$ time obtain from it a $c$-roundtrip spanner of $G$. We note that it is easy to obtain the same result but where both the in- and out-degrees are $O(m/n)$ (see the proof in the appendix).

Now we can assume that the degree of each node is no more than $O(m/n)$. This will allow us for instance to run Dijkstra's algorithm or BFS from a vertex within a neighborhood of $w$ nodes in $\tilde{O}(mw/n)$ time.

Another assumption we can make without loss of generality is that our given graph $G$ is strongly connected. In linear time we can compute the strongly connected components and then run any algorithm on each component separately. We know that any two vertices in different components have infinite roundtrip distance.

A final lemma (implicit in \cite{improvedgirtharx}) will be very important for our algorithms:

\begin{lemma}
Let $G=(V,E)$ be a directed graph with $|V|=n$ and integer edge weights in $\{1,\ldots,M\}$.
Let $S\subseteq V$ with $|S|>c \log n$ (for $c\geq 100/\log(10/9)$) and let $d$ be a positive integer.
Let $R$ be a random sample of $c\log n$ nodes of $S$ and define
$S':=\{s\in S~|~d(s,r)\leq d,~\forall r\in R\}.$
Suppose that for every $s\in S$ there are at most $0.2 |S|$ nodes $v\in V$ so that $d(s,v),d(v,s)\leq d$.
Then $|S'|\leq 0.8 |S|$.
\label{lemma:setreduce}
\end{lemma}

\begin{proof}
The proof will consist of two parts. First we will show that the number of ordered pairs $s,s'\in S$ for which $d(s,s'),d(s',s)\leq d$ is small. Then we will show that if $|S'|>0.8|S|$, then with high probability, the number of ordered pairs $s,s'\in S$ for which $d(s,s'),d(s',s)\leq d$ is large, thus obtaining a contradiction.

(1) If for every $s\in S$ there are at most $0.2 |S|$ nodes $v\in V$ so that $d(s,v),d(v,s)\leq d$, then the number of ordered pairs $s,s'\in S$ for which $d(s,s'),d(s',s)\leq d$ is clearly at most $0.2 |S|^2$.

(2) Suppose now that $|S'|>0.8|S|$. First, consider any $s\in S$ for which there are at least $0.1 |S|$ nodes $s'\in S$ such that $d(s,s')>d$. The probability that $d(s,r)\leq d$ for all $r\in R$ is then at most $0.9^{c\log n}\leq 1/n^{100}.$ Thus, via a union bound, with high probability at least $1-1/n^{99}$, for every $s\in S'$, there are at least $0.9 |S|$ nodes $s'\in S$ such that $d(s,s')\leq d$. 

Now, if $|S'|>0.8|S|$, with high probability, there are at least $0.8|S| \times 0.9|S|=0.72|S|^2$ ordered pairs $(s,s')$ with $s,s'\in S$ and $d(s,s')\leq d$. There are at most ${|S|\choose 2}\leq |S|^2/2$ ordered pairs $(s,s')$ such that exactly one of $\{d(s,s')\leq d,d(s',s)\leq d\}$ holds. Hence, with high probability there are at least $0.22|S|^2>0.2 |S|^2$ ordered pairs $(s,s')$ with $s,s'\in S$ and both $d(s,s')\leq d$ and $d(s',s)\leq d$. Contradiction.
\end{proof}

\section{$2$-Approximation for the Girth in Unweighted Graphs}
Here we show how to obtain a genuine $2$-approximation for the girth in unweighted graphs.
\begin{theorem}
\label{thm:2approx_unw}
Given a directed unweighted graph $G$ on $m$ edges and $n$ nodes, one can in $\tilde{O}(mn^{3/4})$ time compute a $2$-approximation to the girth.
\end{theorem}
Note that this is the first part of Theorem \ref{thm:2approxunw}. The pseudocode for the algorithm of Theorem \ref{thm:2approx_unw} can be found in Algorithm \ref{alg:highgirth}, and we will refer to it at each stage of the proof. 

We will consider two cases for the girth: when it is $\geq n^\delta$ and when  it is $<n^\delta$, for some $\delta>0$ we will eventually set to $1/4$. We will assume that all out-degrees in the graph are $O(m/n)$.

\subsection{Large girth.}
Pick a random sample $R$ of $100 n^{1-\delta} \log n$ nodes, run BFS to and from each $s\in R$. Return 
$$\min_{s\in R}\min_{v\neq s} d(s,v)+d(v,s).$$

If the girth is $\geq n^\delta$, with high probability, $R$ will contain a node $s$ on the shortest cycle $C$. Since any cycle must contain two distinct nodes, $\min_{s\in R}\min_{v\neq s} d(s,v)+d(v,s)$ is the weight of a shortest cycle that contains some node of $R$, and with high probability it must be the girth. Thus in $\tilde{O}(mn^{1-\delta})$ time we have computed the girth exactly. See Procedure $\textsc{HighGirth}$ in Algorithm \ref{alg:highgirth}.

\subsection{Small girth.}
Now let us assume that the girth is at most $n^\delta$. For a vertex $u$ and integer $j\in \{0,\ldots,n^\delta\}$, define 
$$B^{j}(u):=\{x\in V~|~d(u,x)=j\} \textrm{ and } \bar{B}^{j}(u):=\{x\in V~|~d(u,x)\leq j\}.$$

We will try all choices of integers $i$ from $3$ to $n^\delta$ to estimate the girth when it is $\leq i$. 

Our algorithm first computes a random sample $Q$ of size $O(n^{1-t}\log n)$ for a parameter $t$,
does BFS from and to all nodes in $Q$, and computes for each $i\in \{1,\ldots,n^\delta\}$,  $V'_i=\{v\in V~|~\exists q\in Q:~d(v,q)\leq i \textrm{ and } d(q,v)\leq i\}$. 
The running time needed to do this for all $i\leq n^\delta$ is $\tilde{O}(mn^{1-t+\delta})$ \footnote{The running time is actually less, $\tilde{O}(n^{2-t+\delta}+mn^{1-t})$ but this won't matter for our algorithm.}.

If $V'_i\neq\emptyset$, the girth of $G$ must be $\leq 2i$. 

Now, pick the smallest $i$ for which $V'_{i+1}\neq\emptyset$. Then $V'_k=\emptyset$ for all $k\leq i$, and we have certified that the girth is $\leq 2i+2$. If the girth is $\geq i+1$, we already have a $2$-approximation. Otherwise, the girth must be $\leq i$. 

Consider any $u\in V$, and $j\leq i$.
Suppose that for all $j\leq i$, $|B^j(u)|\leq 100n^t$. 
Then, for $u$ and for all $v\in B^j(u)$ for $j\leq i$, we could compute the distances from $u$ to $v$ in $G$ efficiently: We do this by running BFS from $u$ but stopping when a vertex outside of $\cup_{j=0}^i B^j(u)$ is found. Note that the number of vertices in $\cup_{j=0}^i B^j(u)$ is $O(n^t \cdot i)$, 
and since we assumed that the degree of every vertex is $O(m/n)$, we get a total running time of $O(m n^{t-1}\cdot i)$. 
If this works for all vertices $u$, then we would be able to compute all distances up to $i$ exactly in total time
$O(m\cdot i n^t)\leq O(mn^{t+\delta})$. 

Unfortunately, however, some $B^j(u)$ balls can be larger than $100n^t$.
In this case, for every $j\leq i$, we will compute a small set of nodes $B'^j(u)$ that will be just as good as $B^j(u)$ for computing short cycles.

\begin{claim}Fix $i$: $1\leq i\leq n^\delta$.
Suppose that for every $j\leq i$ we are given black box access to sets $B'^j(u)\subseteq \bar{B}_j(u)$ of nodes such that (1) In $t(n)$ time we can check whether a node is in $B'^j(u)$, (2) 
$|B'^j(u)|\leq 100n^t$ whp, and (3) for any cycle $C$ of length $\leq i$ containing $u$, and every $j\leq i$, any node of $C$ that is in $B^j(u)$ is also in $B'^j(u)$.

Then there is an 
$O(mn^{t-1+\delta} t(n))$ time algorithm that can find a shortest cycle through $u$, provided that cycle has length $\leq i$.\label{claim:modbfs}
\end{claim}

\begin{proof}
Let us assume that there is some cycle $C$ of length $\leq i$ containing $u$. Also, assume that we are given the sets $B'^j(u)$ for all $j\leq i$ as in the statement of the lemma.

Then we can compute a modified BFS out of $u$. We will show by induction that when considering distance $j\leq i$, our modified BFS will have found a set $N_j(u)$ of nodes such that for every $x\in N_j(u)$, $d(u,x)\leq j$, and so that for any cycle $C$ of length $\leq i$ containing $u$, any node of $C$ that is in $B^j(u)$ is also in $N_j(u)$. 

Initially, $N_0(u)=\{u\}$, so the base case is fine. Let's make the induction hypothesis for $j$ that for every $x\in N_j(u)$, $d(u,x)\leq j$, and for a shortest cycle $C$ of length $\leq i$ containing $u$, any node of $C$ that is in $B^j(u)$ is also in $N_j(u)$.

Our modified BFS proceeds as follows: Given $N_j(u)$, we go through each $z\in N_j(u)$, and if $z\in B'^j(u)$, we go through all out-neighbors $y$ of $z$, and if $y$ has not been visited until now, we place $y$ into $N_{j+1}(u)$. See Procedure $\textsc{ModBFS}$ in Algorithm \ref{alg:highgirth} (parameter $t$ is set to $1/2$ here).

Clearly, since $d(u,z)\leq j$ (by the induction hypothesis), we have that $d(u,y)\leq j+1$ for each out-neighbor $y$ of $z$. Now consider a shortest cycle $C$ containing $u$ of length $\leq i$. To complete the induction we only care about $j< |C|$.

Assume that the induction hypothesis for $j$ holds. Let $x$ be the node on $C$ at distance $j+1$ from $u$ along $C$, and let $x'$ be its predecessor on $C$, i.e. the node on $C$ at distance $j$ from $u$ along $C$. Since $C$ is a shortest cycle containing $u$ and since $x'\neq x$, we must have that $d(u,x')=j$ so that $x'\in B^j(u)$. Also, either $u=x$, or $d(u,x)=j+1$ and so $x\in B^{j+1}(u)$.

We know by the induction hypothesis that $x'\in N_{j}(u)$ and also that $x'\in B'^j(u)$ by the definition of $B'^j(u)$. Thus, we would have gone through the edges out of $x'$, and $x$ would have been discovered. If $u=x$, then the cycle $C$ will be found. Otherwise,
 $d(u,x)=j+1$, and $x$ cannot have been visited until now, so our modified BFS will insert $x$ into $N_{j+1}(u)$ thus completing the induction.

The running time of the modified BFS is determined by the fact that there are $i\leq n^\delta$ levels, each of $N_j(u)\cap B'^j(u)$ contains 
$\leq O(n^t)$ nodes, and we traverse the $O(m/n)$ edges out of every $x\in N_j(u)\cap B'^j(u)$. The running time is thus asymptotically $t(n)\times n^\delta \times n^t \times m/n$ which is $O(m n^{t+\delta-1} t(n))$.
\end{proof}

Now we want to explain how to compute the sets $B'^j(u)$.
We use Lemma~\ref{lemma:setreduce} from the preliminaries.
Suppose that the girth is at most $i$ and for every $k\leq i$, $V'_k=\emptyset$.

Let $u$ be a node on a cycle $C$ of length at most $i$.
Let $x$ be any node on $C$ so that $x\in B^j(u)$ for some integer $j\leq i$. 
Then we must have that for every $y\in \bar{B}^j(u):$
$$d(x,y)\leq d(x,u)+d(u,y)\leq |C|-d(u,x)+d(u,y) 
\leq i -j+j= i.$$


This inequality is crucial for our algorithm. See Figure~\ref{fig:aroundcycle} for a depiction of it.

\begin{figure}[h]
  \centering
  \includegraphics[width=.25\linewidth]{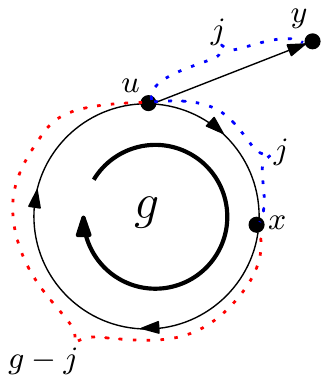}
  \caption{Here there is a cycle of length $g$ containing $u$. A node $x$ on the cycle is at distance $j$ from $u$ along the cycle and another node $y$ is at distance $\leq j$ from $u$. Then the distance from $x$ to $y$ is at most $g$ since one way to go from $x$ to $y$ is to go from $x$ to $u$ along the cycle at a cost of $g-j$, and then from $u$ to $y$ at a cost of $\leq j$. If the cycle is a shortest cycle containing $u$ and if $x\neq u$, then the distance in the graph from $u$ to $x$ is $j$, as the path along the cycle needs to be a shortest path.}
  \label{fig:aroundcycle}
\end{figure}

In other words, we obtain that $x$ is in $\{w\in B^j(u)~|~d(w,y)\leq i,~ \forall y\in \bar{B}^j(u)\}$.

Suppose that we are able to pick a random sample $R^j(u)$ of $c\log n$ vertices from $\bar{B}^j(u)$ (we will show how later). Then we can define 
$$\bar{B}'^j(u):=\{z\in \bar{B}^j(u)~|~d(z,y)\leq i,~ \forall y\in R^j(u)\}.$$

Using Lemma~\ref{lemma:setreduce} we will show that if $|\bar{B}^j(u)|\geq 10n^t$, 
then $|\bar{B}'^j(u)|\leq 0.8 \bar{B}^j(u)$ and if $x$ is in $\{w\in B^j(u)~|~d(w,y)\leq i,~ \forall y\in \bar{B}^j(u)\}$, then whp $x\in \bar{B}'^j(u)$. We will then repeat the argument to obtain $B'^j(u)$ of size $O(n^t)$.

Consider any $s\in V$ with at least $0.2 |\bar{B}^j(u)|$
 nodes $v\in V$ so that $d(s,v),d(v,s)\leq i$. As 
%
$|\bar{B}^j(u)|\geq 10n^t$ (as otherwise we would be done), $0.2 |\bar{B}^j(u)|\geq 2n^t$, 
and so with high probability, for $s$ with the property above, our earlier random sample $Q$ contains some $q$ with $d(s,q),d(q,s)\leq i$, and so $V'_i\neq \emptyset$ which we assumed didn't happen. Thus with high probability, for every $s\in V$, there are at most $0.2 |\bar{B}^j(u)|$
 nodes $v\in V$ so that $d(s,v),d(v,s)\leq i$.
Hence we also have that every $z\in \bar{B}^j(u)$ has at most $0.2 |\bar{B}^j(u)|$
 nodes $v\in V$ so that $d(s,v),d(v,s)\leq i$.
 
Thus we can apply Lemma~\ref{lemma:setreduce} to $\bar{B}^j(u)$ and conclude that $|\bar{B}'^j(u)|\leq 0.8 |\bar{B}^j(u)|$, while also any node $x\in B^j(u)$ on the cycle $C$ (containing $u$) is also in $\bar{B}'^j(u)$. 

 We will iterate this process until we arrive at a subset of $\bar{B}^j(u)$ that is smaller than $10n^t$ 
and still contains all $x\in B^j(u)$ on an $\leq i$-length cycle $C$.

We do this as follows. Let $B^j_{0}(u)=\bar{B}^j(u)$. For each $k=0,\ldots, 2\log n$, let $R^j_{k}(u)$ be a random sample of $O(\log n)$ vertices of $B^j_{k}(u)$. Define $B^j_{k+1}(u)=\{z\in B^j_{k}(u)~|~d(z,y)\leq i,~ \forall y\in \cup_{\ell=0}^k R^j_{\ell}(u)\}$. 
We get that for each $k$, $|B^j_{k}(u)|\leq 0.8^k |\bar{B}^j(u)|$ so that at the end of the last iteration, 
$|B^j_{2\log n}(u)|\leq 10n^t$ and we can set $B'^j(u)$ to $B^j_{2\log n}(u)$.

It is not immediately clear how to obtain the random sample $R^j_{k}(u)$ from $B^j_{k}(u)$ as $B^j_{k}(u)$ is unknown. We do it in the following way, adapting an argument from Chechik et al. \cite{improvedgirthSTOC}. For each $j\leq i$ and $k\leq 2\log n$ we independently obtain a random sample $S_{j,k}$ of $V$ by sampling each vertex independently with probability $p=100\log n / n^t$.
For each of the (in expectation) $O(n^{1-t+\delta}  \log^2(n))$ 
vertices in the sets $S_{j,k}$ we run BFS to and from them, to obtain all their distances.

Now, for $j\leq i$ and $k$, to obtain the random sample $R^j_{k}(u)$ of the unknown $B^j_{k}(u)$, we assume that we already have $R^j_{\ell}(u)$ for $\ell<k$, and define
$$T^j_{k}(u)=\{s\in S_{j,k}~|~s\in \bar{B}^j(u) \textrm{ and } d(s,y)\leq i,~ \forall y\in \cup_{\ell<k} R^j_{\ell}(u)\}.$$
Forming the set $T^j_{k}(u)$ is easy since we have the distances $d(s,v)$ for all $s\in S_{j,k}$ and $v\in V$, so we can check whether $s\in \bar{B}^j(u)$ and $d(s,y)\leq i,~ \forall y\in \cup_{\ell<k} R^j_{\ell}(u)$ in polylogarithmic time for each $s\in S_{j,k}$.  See Procedure $\textsc{RandomSamples}$ in Algorithm \ref{alg:highgirth}.

Now since $S_{j,k}$ is independent from all our other random choices, $T^j_{k}(u)$ is  a random sample of $B^j_{k}(u)$ essentially created by selecting each vertex with probability $p$. If 
$B^j_{k}(u)\geq 100n^t$,
 with high probability, $T^{j}_{k}(u)$ has at least $10\log n$ vertices so we can pick $R^j_{k}(u)$ to be a random sample of $10\log n$ vertices of $T^j_{k}(u)$, and they will also be a random sample of $10\log n$ vertices of $B^j_{k}(u)$.

Once we have the sets $R^j_{k}(u)$ for each $u$ and $j\leq i$, $k\leq 2\log n$, we run our modified BFS from each $u$ from Claim~\ref{claim:modbfs}
 where when we are going through the vertices $x\in N_j(u)$
 we check whether $x\in B'^j(u)$ by checking whether
  $d(x,r)\leq i$ for every $r\in \cup_k R^j_{k}(u)$. This only gives a polylogarithmic overhead so we can run the modified BFS in time 
  $\tilde{O}(mn^{t-1+\delta})$ time.
  We can run it through all $u\in V$ in total time   $\tilde{O}(mn^{t+\delta})$ time,
  and in this time we will be able to compute the length of the shortest cycle if that cycle is of length $\leq i$.

\paragraph{Putting it all together.}
In $\tilde{O}(mn^{1-\delta})$ time we compute the girth exactly if it is $\geq n^\delta$. In $\tilde{O}(mn^{1-t+\delta})$ time, we obtain $i$ so that we have a
$2$-approximation of the girth if the girth is $>i$. In additional $\tilde{O}(mn^{1-t+\delta}+mn^{t+\delta})$ time we compute the girth exactly if it is $\leq i$.

To optimize the running time we set $t=1/2$, $1-\delta=0.5+\delta$, obtaining $\delta=1/4$, and a running time of $\tilde{O}(mn^{3/4})$.
The final algorithm is in Algorithm \ref{alg:highgirth}.
\begin{algorithm}
{
\fontsize{10}{10}\selectfont
\caption{$2$-Approximation algorithm for the girth in unweighted graphs.}
\label{alg:highgirth}
\SetKwProg{procedure}{Procedure}{}{}
\procedure{$\textsc{HighGirth}(G=(V,E))$}{
    Let $R\subseteq V$ be a uniform random sample of $100n^{3/4}\log n$ nodes.\; 
    \ForEach{$s\in R$}{
            Do BFS from $s$ in $G$\;
    }
    Let $g$ be the length of the shortest cycle found by the BFS searches.\;
    Return $g$.
}
\label{alg:randsampunw}
\SetKwProg{procedure}{Procedure}{}{}
\procedure{$\textsc{RandomSamples}(G=(v,E),i)$}{
    \ForEach{$j\in \{1,\ldots,i\}$}{
        \ForEach{$k\in \{1,\ldots,2\log n\}$}{
            Let $S_{j,k}\subseteq V$ be a uniform random sample of $100\sqrt n \log n$ vertices.\;
            \ForEach{$s\in S_{j,k}$}{
                Do BFS to and from $s$ to compute for all $v$, $d(s,v)$ and $d(v,s)$.\;
            }
        }  
    }
    \ForEach{$u\in V$}{
        \ForEach{$j\in \{1,\ldots,i\}$}{
            $R^j(u)\leftarrow \emptyset$.\;
            \ForEach{$k\in \{1,\ldots,2\log n\}$}{
                $T_k^j(u)\leftarrow \{s\in S_{j,k}~|~d(u,s)\leq j \textrm{ and for all } y\in R^j(u):~d(s,y)\leq i\}$.\;
            
                \If{$|T_k^j(u)|< 10\log n$}{
                    $R^j(u)\leftarrow R^j(u)\cup T_k^j(u)$\;
                    Exit this loop (over $k$).
                 }
                \Else{
                    Let $R_k^j(u)$ be a uniform random sample of $10\log n$ nodes from $T_k^j(u)$.\;
                    $R^j(u)\leftarrow R^j(u)\cup R_k^j(u)$.\;
                }
            }
        }
    } Return the sets $R^j(u)$ for all $j\leq i$, $u\in V$, and $d(s,v),d(v,s)$ for all $s\in \cup_{j,k}S_{j,k}$ and $v\in V$.
}

\label{alg:modbfs}
\SetKwProg{procedure}{Procedure}{}{}
\procedure{$\textsc{ModBFS}(G=(v,E),u,i,R^1(u),\ldots,R^i(u)), d(\cdot)$}{
    // $d(\cdot)$ contains $d(s,v),d(v,s)$ for all $s\in \cup_{j,k}S_{j,k}$ and $v\in V$.\;

    $Visited\leftarrow$ empty hash table\;
    $N_0\leftarrow \{u\}$\;
    $Visited.insert(u)$\;
    \ForEach{$j$ from $0$ to $i-1$}{
        $N_{j+1}\leftarrow$ empty linked list\;
        \ForEach{$x\in N_j$}{
            \If{for every $s\in R^j(u)$, $d(x,s)\leq i$}{
                \ForEach{$y$ s.t. $(x,y)\in E$ and $y\notin Visited$}{
                    \If{$y=u$}{
                        Stop and return $j+1$\;
                    }
                    $N_{j+1}.insert(y)$\;
                    $Visited.insert(y)$\;
                }
            }
        }
    }
    Return $\infty$ // No $\leq i$ length cycle found through $u$\;
}    

\label{alg:girthapp}
\SetKwProg{procedure}{Procedure}{}{}
\procedure{$\textsc{GirthApprox}(G=(V,E))$}{
    $g_{high}\leftarrow \textsc{HighGirth}(G)$\;
    
    Let $Q\subseteq V$ be a uniform random sample of $100n^{1/2}\log n$ nodes.\; 
    \ForEach{$s\in Q$}{
            Do BFS from and to $s$ in $G$\;
    }
    Let $i$ be the minimum integer s.t. $\exists s\in Q$ and $\exists v\in V$ with $d(s,v)\leq i+1$ and $d(v,s)\leq i+1$.\;
    $g_{med}\leftarrow 2(i+1)$\;
    
    Let $i$ be the min of $i$ and $n^{1/4}$\;
    
    Run $\textsc{RandomSamples}(G,i)$ to obtain sets $R^j(u)$ for all $j\leq i$, $u\in V$, and $d(\cdot)$ containing $d(s,v),d(v,s)$ for all $s\in \cup_{j,k}S_{j,k}$ and $v\in V$\;
    
    \ForEach{$u\in V$}{
        $g_u\leftarrow \textsc{ModBFS}(G,u,i,R^1(u),\ldots,R^i(u), d(\cdot))$\;
    }
    
    $g\leftarrow \min\{g_{high},g_{med}, \min_{u\in V} g_u\}$\;
    Return $g$\;
}
}
\end{algorithm}

\section{Weighted Graphs: Girth and Roundtrip Spanner.}
\label{sec:weighted}
One of the main differences between our weighted and unweighted algorithms is that for weighted graphs we do not go through each distance value $i$ up to $n^\delta$, but we instead 
process intervals of possible distance values $[(1+\eps)^i,(1+\eps)^{i+1})$ for small $\eps>0$.
This will affect the approximation, so that we will get a $(2+O(\eps))$-approximation. However, it will also enable us to have a smaller running time of $\tilde{O}(m\sqrt n \log(M) /\eps)$, and to be able to output an $\tilde{O}(n^{1.5}\log(M)/\eps)$-edge $(5+O(\eps))$-approximate roundtrip spanner in $\tilde{O}(m\sqrt n \log(M) /\eps^2)$ time, where $M$ is the maximum edge weight. 

Fix $\eps>0$.
For a vertex $u$ and integer $j$, define (differently from the previous section)
$$B^{j}(u):=\{x\in V~|~(1+\eps)^j\leq d(u,x)< (1+\eps)^{j+1}\} \textrm{ and } \bar{B}^{j}(u):=\{x\in V~|~ d(u,x)< (1+\eps)^{j+1}\}.$$
We include a boundary case $B^{\emptyset}(u):=\{x\in V~|~d(u,x)=0\}$. Recall that we originally started with a graph with positive integer weights, but our transformation to vertices of degree $O(m/n)$ created some $0$ weight edges. We note that any distance of $0$ involves at least one of the auxiliary vertices and no roundtrip distance can be $0$.

In our algorithms including our $(2+\epsilon)$-approximation algorithm, we do a restricted version of Dijkstra from every vertex where before running these Dijkstras, we need to efficiently sample a set of vertices $R^j(u)$ of size $O(\log{n})$ from a subset of $B^j(u)$, without computing the set $B^j(u)$. The following lemma is given as input the target approximation factor $2\beta$, a parameter $i$ as an estimated size of cycles the algorithm is handling at a given stage and a parameter $\alpha$ as the target running time $\tilde{O}(mn^\alpha)$ of our algorithms. It outputs the sample sets in this running time. The proof of the lemma is similar to the sampling method of the previous section and is included in the Appendix.

\begin{lemma}
\label{lemma:modifiedDijkstra}
Let $M$ be the maximum edge weight of the graph and suppose that $i\in \{1,\ldots,\log_{1+\epsilon}Mn\}$, $\beta>0$ and $0<\alpha<1$ are given. Suppose that $Q$ is a given sampled set of size $\tilde{O}(n^{\alpha})$ vertices. Let $d=\beta (1+\epsilon)^{i+1}$.  Let $V'_i=\{v\in V~|~\exists q\in Q:~d(v,q)\leq d \textrm{ and } d(q,v)\leq d\}$. In $\tilde{O}(mn^\alpha)$ time, for every $u\in V$ and every $j=\{1,\ldots, \log_{(1+\epsilon)}(Mn)\}$, one can output a sample set $R^j_{i}(u)$ of size $O(\log^2{n})$ from $\bar{Z}^j_i(u)=\bar{B}^j(u)\setminus V'_i$, where the number of vertices in ${Z}^j_i(u)=B^j(u)\setminus V'_i$ of distance at most $d$ from all vertices in $R^j_{i}(u)$ is at most $O(n^{1-\alpha})$ whp.
\end{lemma}

Now we focus on our $(2+O(\epsilon))$-approximation algorithm for the girth and $(5+O(\epsilon))$-approximate roundtrip spanner. We are going to prove the following Theorem, which consists of Theorem \ref{thn:weighted_spanner} and the second part of Theorem \ref{thm:2approxunw} with a $\log{M}$ factor added to their running times.  
\begin{theorem} 
\label{thm:2approx}
Let $G$ be an $n$-node, $m$-edge directed graph with edge weights in $\{1,\ldots, M\}$. Let $\eps>0$.
One can compute a $(5+\eps)$-roundtrip spanner on $\tilde{O}(n^{1.5}\log^2 M/\eps^2)$ edges in $\tilde{O}(m\sqrt n \log^2(M)/\eps^2)$ time, whp. In $\tilde{O}(m\sqrt n \log(M)/\eps)$ time, whp, one can compute a $(2+\eps)$-approximation to the girth.
\end{theorem}

We will start with a sampling approach, similar to that in the unweighted girth approximation. The pseudocode of the girth algorithm can be found in Algorithm \ref{alg:randsamp}, and we will refer to it at each stage of the proof.

\begin{lemma}
Let $G=(V,E)$ be a directed graph with $|V|=n$ and integer edge weights in $\{1,\ldots,M\}$.
Let $d$ be a positive integer, $\eps\geq 0$, and let $Q\subseteq V$ be a random sample of $100\sqrt n \log n$ vertices. In $\tilde{O}(m\sqrt n)$ time we can compute shortest paths trees $T^{in}(q), T^{out}(q)$ into and out of each $q\in Q$.
Let $H$ be the subgraph of $G$ consisting of the edges of these trees $T^{in}(q), T^{out}(q)$. Let $V'=\{v\in V~|~\exists q\in Q,~d(v,q)\leq d \textrm{ and } d(q,v)\leq d\}$.
Then: \begin{itemize}
\item {\bf Girth approximation:} If $V'\neq \emptyset$, then the girth of $G$ is at most $2d$.
\item {\bf Additive distance approximation:} For any $u,v\in V$, if the shortest $u$ to $v$ path contains a node of $V'$, then $d_H(u,v)\leq d(u,v)+2d$.
\item {\bf Sparsity:} The number of edges in $H$ is $\tilde{O}(n^{1.5})$.
\end{itemize}
\label{lemma:hitpath}
\end{lemma} 

\begin{proof}
Given a directed $G=(V,E)$ with $|V|=n$, $|E|=m$ and edge weights in $\{1,\ldots,M\}$, let us first take a random sample $Q\subseteq V$ of $100\sqrt n \log n$ vertices. Run Dijkstra's algorithm from and to every $q\in Q$. Determine $V'\subseteq V$ defined as those $v\in V$ for which there is some $q\in Q$ with $d(v,q),d(q,v)\leq d.$
If $V'\neq \emptyset$, we get that the girth of $G$ is at most $2d$.
Suppose that we insert all edges of the in- and out- shortest paths trees rooted at all $q\in Q$ into a subgraph $H$. Then we have only inserted $\tilde{O}(n^{1.5})$ edges as each tree has $\leq n-1$ edges. 

Consider some $u,v\in V$ such that there is some node $x\in V'$ on the shortest $u-v$ path. Let $q\in Q$ be such that $d(x,q),d(q,x)\leq d$. Then $$d_H(u,v)\leq d(u,q)+d(q,v)\leq d(u,x)+d(x,q)+d(q,x)+d(x,v)\leq d(u,v)+2d.$$

\end{proof}

Our approach below will handle the roundtrip spanner and the girth approximation at the same time. 

We will try all choices of integers $i$ from $0$ to $\log_{1+\eps}(Mn)$ to estimate roundtrip distances in the interval $[(1+\eps)^i,(1+\eps)^{i+1})$, and to estimate the girth if it is $< (1+\eps)^{i+1}$. 

Fix a choice for $i$ for now.

Our algorithm first applies the approach of Lemma~\ref{lemma:hitpath} by setting $d=(1+\eps)^{i+2}$ (we will see later why). We compute a random sample $Q$ of size $O(\sqrt n \log n)$, do Dijkstra's from and to all nodes in $Q$, and add the edges of the computed shortest paths trees to our roundtrip spanner $H$. We also compute 
$$V'_i=\{v\in V~|~\exists q\in Q:~d(v,q)\leq (1+\eps)^{i+2} \textrm{ and } d(q,v)\leq (1+\eps)^{i+2}\}.$$ 

By Lemma~\ref{lemma:hitpath}, if $V'_i\neq\emptyset$, the girth of $G$ must be $\leq 2(1+\eps)^{i+2}$. For the choice of $i$ where $(1+\eps)^i\leq g\leq (1+\eps)^{i+1}$, we will get an approximation factor of $2(1+\eps)^2\leq 2(1+3\eps)$.
Just as with the algorithm for unweighted graphs, we can pick the minimum $i$ so that $V'_{i}\neq \emptyset$, use $2(1+\eps)^{i+2}$ as one of our girth estimates and then proceed from now on with a single value $i-1$ considering only the interval $[(1+\eps)^{i-1},(1+\eps)^{i})$.

By Lemma~\ref{lemma:hitpath}, we also get that for any $u,v\in V$  for which the $u$-$v$ shortest path contains a node of $V'_i$, $H$ gives a good additive estimate of $d(u,v)$, i.e. $d(u,v)\leq d_H(u,v)\leq d(u,v)+2(1+\eps)^{i+2}.$

Suppose that also $(1+\eps)^i\leq \rd{u}{v}\leq (1+\eps)^{i+1}$, and that we somehow also get a good estimate for $d(v,u)$ (either because the $v$-$u$ shortest path contains a node of $V'$, or by adding more edges to $H$), so that also 
$d(v,u)\leq d_H(v,u)\leq d(v,u)+2(1+\eps)^{i+2}.$
Then,
\[\rd{u}{v}\leq \rdh{H}{u}{v}\leq \rd{u}{v}+4(1+\eps)^{i+2}\leq \rd{u}{v}(1+4(1+3\eps))\leq  \rd{u}{v}(5+12\eps).\]

In other words, we would get a $5+O(\eps)$-roundtrip spanner, as long as by adding $\tilde{O}(n^{1.5})$ edges to $H$, we can get a good additive approximation to the weights of the $u$-$v$ shortest paths that do not contain nodes of $V'_i$, for all $u,v$ with $(1+\eps)^i\leq \rd{u}{v}\leq (1+\eps)^{i+1}$. We will in fact compute these shortest paths exactly. For the girth $g$ itself, we will show how to compute it exactly, if no node of $V'_i$ hit the shortest cycle, where $i$ is such that $(1+\eps)^{i-1}\leq g\leq (1+\eps)^i$.


Fix $i$.
Let $Z_i=V\setminus V'_i$ and $d=(1+\eps)^{i+1}$. We can focus on the subgraph induced by $Z_i$.


Consider any $u\in Z_i$, and $j\leq i$.
Define $Z^j_i(u)=Z_i\cap B^j(u)$ and $\bar{Z}^j_i(u)=Z_i\cap \bar{B}^j(u)$. We also add the boundary case $Z^{\emptyset}_i=Z_i\cap B^{\emptyset}(u)=\{x\in Z_i~|~d(u,x)=0\}$.


If for all $j\in \{\emptyset\}\cup \{1,\ldots,i\}$, $|Z^j_i(u)|\leq 100\sqrt n$, running Dijkstra's algorithm from $u$ in the graph induced by $Z_i$, up to distance $(1+\eps)^{i+1}$ would be cheap.
Unfortunately, however, some $Z_i^j(u)$ balls can be larger than $100\sqrt n$. In this case, similarly to our approach for the unweighted case, we will replace $Z_i^j(u)$ with a set $Z_i^{\prime j}(u)\subseteq \bar{Z}_i^j(u)$ of size $O(\sqrt n)$ with the guarantee that 
for any $v\in V$ with $(1+\eps)^i\leq \rd{u}{v}< (1+\eps)^{i+1}$ for which the shortest $u$-$v$ path does not contain a node of $V'_i$, every node of this $u$-$v$ shortest path that is in $Z_i^j(u)$ is also in $Z_i^{\prime j}(u)$. 

The following lemma shows how to use such replacement sets.

\begin{lemma} Let $u$ and $i$ be fixed.
Suppose that for every $j\in \{\emptyset\}\cup \{1,\ldots,i\}$ we are given black box access to sets $Z_i^{\prime j}(u)\subseteq \bar{Z}_i^j(u)$ of nodes such that (1) Checking whether a node $z$ is in $Z_i^{\prime j}(u)$ takes $t(n)$ time, (2) $|Z^{\prime j}_i(u)|\leq 100\sqrt n$ whp, and (3) for any $v$ such that $(1+\eps)^i\leq \rd{u}{v}\leq (1+\eps)^{i+1}$, and every $j\leq i$, every node on the shortest path $P$ from $u$ to $v$ that is in $Z_i^j(u)$ is also in $Z_i^{\prime j}(u)$.

Then there is an $\tilde{O}(m \log(M) t(n)/(\eps\sqrt n))$ time algorithm that finds a shortest path from $u$ to any $v$ with $(1+\eps)^i\leq \rd{u}{v}\leq (1+\eps)^{i+1}$ and s.t. the shortest $u$-$v$ path does not contain a node of $V_i'$. The algorithm returns $\tilde{O}(n^{0.5}\log(M)/\eps)$ edges whose union contains all these shortest paths.\label{lemma:moddijk}
\end{lemma}

\begin{proof}
Assume we have the sets $Z_i^{\prime j}(u)$ for $j\in \{\emptyset\}\cup \{1,\ldots,i\}$ as in the statement of the lemma.

Then we will define a modified Dijkstra's algorithm out of $u$. The algorithm begins by placing $u$ in the Fibonacci heap with $d[u]=0$ and all other vertices with $d[\cdot]=\infty$. When a vertex $x$ is extracted from the heap with estimate $d[x]$, we determine the $j$ for which $(1+\eps)^j\leq d[x]<(1+\eps)^{j+1}$; here $j$ could be the boundary case that we called $\emptyset$ if $d[x]=0$.
Then we check whether $x$ is in $Z_i^{\prime j}$. If it is not, we ignore it and extract a new vertex from the heap. Otherwise if $x\in Z_i^{\prime j}$, we go through all its out-edges $(x,y)$, and if $d[y]>d[x]+w(x,y)$, we update $d[y]=d[x]+w(x,y)$. 
For any new cycle to $u$ found, we update the best weight found, and in the end we return it. See Procedure $\textsc{ModDijkstra}$ in Algorithm \ref{alg:randsamp}.

Since we only go through the edges of at most $O(\sqrt n \log (Mn) / \eps)$ vertices and the degrees are all $O(m/n)$, the runtime is $O(m\log (Mn) / (\eps \sqrt n))$. For the same reason, the modified shortest paths tree whose edges we add to our roundtrip spanner has at most $O(\sqrt n \log (Mn) / \eps)$ edges.

Let $v$ be such that $(1+\eps)^i\leq \rd{u}{v}\leq (1+\eps)^{i+1}$ and for which the shortest $u$-$v$ path does not contain a node of $V'_i$. We will show by induction that our modified Dijkstra's algorithm will compute the shortest path from $u$ to $v$ exactly. 

The induction will be on the distance from $u$.
Let's call the nodes on the shortest $u$ to $v$ path, $u=u_0,u_1,\ldots,u_t=v$. The induction hypothesis for $u_k$ is that $u_k$ is extracted from the heap with $d[u_k]=d(u,u_k)$. Let us show that $u_{k+1}$ will also be extracted from the heap with $d[u_{k+1}]=d(u,u_{k+1})$. The base case is clear since $u$ is extracted first.

When $u_k$ is extracted from the heap, by the induction hypothesis, $d[u_k]=d(u,u_k)$. Let $j$ be such that $(1+\eps)^j\leq d[u_k]<(1+\eps)^{j+1}$. As no node on the $u$-$v$ shortest path is in $V_i'$, we get that $u_k\in Z_i^j$. By the assumptions in the lemma, we also have that $u_k\in Z_i^{\prime j}$. Thus, when $u_k$ is extracted, we will go over its edges. In particular, $(u_k,u_{k+1})$ will be scanned, and $d[u_{k+1}]$ will be set to (or it already was) $d[u_k]+w(u_k,u_{k+1})=d(u,u_{k+1})$.
This completes the induction.

It is also not hard to see that the girth will be computed exactly if $u$ is on a shortest cycle, the girth is in $[(1+\eps)^i,(1+\eps)^{i+1})$ and $V'_i$ is empty.
\end{proof}

Now we compute the sets $Z'^j_i(u)$. 
First consider $u,v\in V$ with $(1+\eps)^i\leq \rd{u}{v}< (1+\eps)^{i+1}$.
Let $x$ be any node on the $u$ to $v$ roundtrip path (cycle) so that $x\in Z_i^j(u)$ for some integer $j\leq i$.
Recall that this means $(1+\eps)^j\leq d(u,x)<(1+\eps)^{j+1}$.
Then for every $y$ with $d(u,y)<(1+\eps)^{j+1}$ and so for each $y\in \bar{Z}_i^j(u)$ we must have (see Figure \ref{fig:aroundcyclewt}) that 
\[d(x,y)\leq d(x,u)+d(u,y)\leq \rd{u}{v}-d(u,x)+d(u,y) \leq \rd{u}{v}-(1+\eps)^j+(1+\eps)^{j+1}\]
\[=
\rd{u}{v}+\eps(1+\eps)^{j}\leq \rd{u}{v}+\eps(1+\eps)^{i}\leq \rd{u}{v}(1+\eps)\leq
(1+\eps)^{i+2}.\]

\begin{figure}[h]
  \centering
  \includegraphics[width=.45\linewidth]{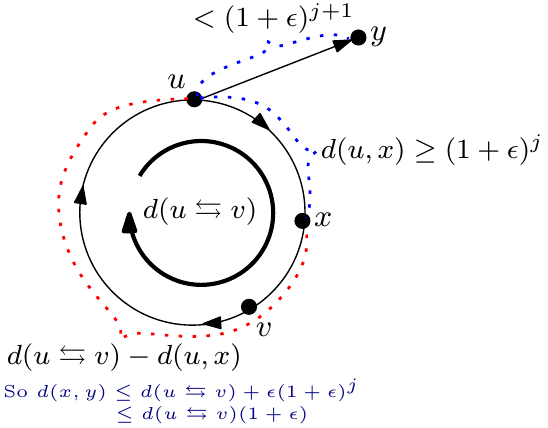}
  \caption{Here $u$ and $v$ have roundtrip distance more than $(1+\eps)^j$. A node $x$ on the shortest $u$-$v$ path is at distance at least $(1+\eps)^j$ from $u$, 
 and another node $y$ is at distance at most $(1+\eps)^{j+1}$ from $u$. Then the distance from $x$ to $y$ is at most $\rd{u}{v}(1+\eps)$ since one way to go from $x$ to $y$ is to go from $x$ to $u$ along the $u$-$v$ roundtrip cycle at a cost of at most $\rd{u}{v}-(1+\eps)^j$, and then from $u$ to $y$ at a cost of at most $(1+\eps)^{j+1}$.}
  \label{fig:aroundcyclewt}
\end{figure}

In other words, $x$ must be in $\{w\in \bar{Z}_i^j(u)~|~d(w,y)\leq (1+\eps)^{i+2} ,~\forall y\in \bar{Z}_i^j(u)\}$. 

We apply Lemma \ref{lemma:modifiedDijkstra} for $\beta = (1+\epsilon)$ and $\alpha = 1/2$. It outputs sets $R_i^j(u)$ of size $O(\log^2{n})$ vertices, where the number of vertices in $\bar{Z}_i^j(u)$ that are at distance $(1+\epsilon)^{i+2}$ from all vertices in $R_i^j(u)$ is $O(\sqrt{n})$ (See Procedure $\textsc{RandomSamplesWt}$ in Algorithm \ref{alg:randsamp}). So all vertices $x\in Z_i^j(u)$ that are in a roundtrip path $u-v$ with $(1+\eps)^i\leq \rd{u}{v}< (1+\eps)^{i+1}$ are in this set, so we let $Z_i^{\prime j}(u) = \{w\in \bar{Z}_i^j(u)| d(w,y)\le (1+\epsilon)^{i+2},~\forall y\in R_i^j(u)\}$.

Now that we have the random samples, we implement the modified Dijkstra's algorithm from Lemma~\ref{lemma:moddijk} with only a polylogarithmic overhead as follows: 

Fix some $j$. Let's look at the vertices $x$ with $(1+\eps)^j\leq d[x]<(1+\eps)^{j+1}$ that the modified Dijkstra's algorithm extracts from the heap. Since $d[x]$ is always an overestimate, $d(u,x)\leq d[x]<(1+\eps)^{j+1}$, and so $x\in \bar{B}^j(u)$. Now, since $x$ is already in $\bar{B}^j(u)$, to check whether $x\in Z'^j_i(u)$, we only need to check whether $x\in Z_i$ (easy) and whether $d(x,y)\leq (1+\eps)^{i+2}$ for all $y\in R^j_{i}(u)$ (this takes $O(\log^2 n)$ time since we have all the distances to the nodes in the random samples).

The final running time is $\tilde{O}(m\sqrt n \log^2(M)/\eps^2)$ since we need to run the above procedure $O(\log(Mn)/\eps)$ times, once for each $i$, and each procedure costs $\tilde{O}(m\log(M)\sqrt n/\eps)$ time.
As we mentioned before, to estimate the girth to within a $(2+\eps)$-factor, we do not need to run the procedure for all $i$ but (as with the algorithm for unweighted graphs), only for the minimum $i$ for which $V'_{i+1}\neq \emptyset$. Thus the running time for the girth becomes  $\tilde{O}(m\sqrt n \log(M)/\eps)$. See Procedure $\textsc{GirthApproxWt}$ in Algorithm \ref{alg:randsamp}.
\begin{algorithm}
{
\fontsize{10}{10}\selectfont
\caption{$2+\eps$-Approximation algorithm for the girth in weighted graphs.}
\label{alg:randsamp}
\SetKwProg{procedure}{Procedure}{}{}
\procedure{$\textsc{RandomSamplesWt}(G=(v,E),i,\eps)$}{
    \ForEach{$j\in \{1,\ldots,i\}$}{
        \ForEach{$k\in \{1,\ldots,2\log n\}$}{
            Let $S_{j,k}\subseteq V$ be a uniform random sample of $100\sqrt n \log n$ vertices.\;
            \ForEach{$s\in S_{j,k}$}{
                Run Dijkstra's to and from $s$ to compute for all $v$, $d(s,v)$ and $d(v,s)$.\;
            }
        }  
    }
    \ForEach{$u\in V$}{
        \ForEach{$j\in \{0,\ldots,i\}$}{
            $R^j(u)\leftarrow \emptyset$.\;
            \ForEach{$k\in \{1,\ldots,2\log n\}$}{
                $T_k^j(u)\leftarrow \{s\in S_{j,k}~|~d(u,s)< (1+\eps)^{j+1} \textrm{ and for all } y\in R^j(u):~d(s,y)\leq (1+\eps)^{i+2}\}$.\;
            
                \If{$|T_k^j(u)|< 10\log n$}{
                    $R^j(u)\leftarrow R^j(u)\cup T_k^j(u)$\;
                    Exit this loop (over $k$).
                 }
                \Else{
                    Let $R_k^j(u)$ be a uniform random sample of $10\log n$ nodes from $T_k^j(u)$.\;
                    $R^j(u)\leftarrow R^j(u)\cup R_k^j(u)$.\;
                }
            }
        }
    } Return the sets $R^j(u)$ for all $j\leq i$, $u\in V$, and $d(s,v),d(v,s)$ for all $s\in \cup_{j,k}S_{j,k}$ and $v\in V$.
}

\label{alg:moddijk}
\SetKwProg{procedure}{Procedure}{}{}
\procedure{$\textsc{ModDijkstra}(G=(v,E),u,i,\eps,R^1(u),\ldots,R^i(u)), d(\cdot)$}{
    // $d(\cdot)$ contains $d(s,v),d(v,s)$ for all $s\in \cup_{j,k}S_{j,k}$ and $v\in V$.\;

    $F\leftarrow$ empty Fibonacci heap\;
    $Extracted\leftarrow$ empty hash table\;
    $F.insert(u, 0)$\;
    $g_u\leftarrow\infty$\;
    \While{$F$ is nonempty}{
        $(x,d[x])\leftarrow F.extractmin$\;
        $Extracted.insert(x)$\;
        \If{for every $s\in R^j(u)$, $d(x,s)\leq (1+\eps)^{i+2}$}{
            \ForEach{$y$ s.t. $(x,y)\in E$}{
                \If{$y\notin Extracted$}{
                    \If{$y$ is in $F$}{
                        $F.DecreaseKey(y,d[x]+w(x,y))$\;
                    }
                    \Else{
                        $F.insert(y,d[x]+w(x,y))$\;
                    }
                }
                \If{$y=u$}{
                    $g_u\leftarrow \min\{g_u,d[x]+w(x,y)\}$\;
                }
            }
        }
    }
    
    Return $g_u$\;
}    

\label{alg:girthapp2}
\SetKwProg{procedure}{Procedure}{}{}
\procedure{$\textsc{GirthApproxWt}(G=(V,E),\eps)$}{  
    Let $Q\subseteq V$ be a uniform random sample of $100n^{1/2}\log n$ nodes.\; 
    \ForEach{$s\in Q$}{
            Do Dijkstra's from and to $s$ in $G$\;
    }
    Let $i$ be the minimum integer s.t. $\exists s\in Q$ and $\exists v\in V$ with $d(s,v)< (1+\eps)^{i+2}$ and $d(v,s)< (1+\eps)^{i+2}$.\;
    $g_{med}\leftarrow \min_{s\in Q, v\in V} d(s,v)+d(v,s)$ // $g_{med}<2(1+\eps)^{i+2}$\;
    \If{$i<0$}{
        // Here $i=-1$ and $d(s,v)=d(v,s)=1$\;
        Return $g_{med}$\;
    }
    
    Run $\textsc{RandomSamplesWt}(G,i,\eps)$ to obtain sets $R^j(u)$ for all $j\leq i$, $u\in V$, and $d(\cdot)$ containing $d(s,v),d(v,s)$ for all $s\in \cup_{j,k}S_{j,k}$ and $v\in V$\;
    
    \ForEach{$u\in V$}{
        $g_u\leftarrow \textsc{ModDijkstra}(G,u,i,\eps,R^1(u),\ldots,R^i(u), d(\cdot))$\;
    }
    
    $g\leftarrow \min\{g_{med}, \min_{u\in V} g_u\}$\;
    Return $g$\;
}
}
\end{algorithm}

\subsection{$(4+\epsilon)$-Approximation Algorithm for the Girth in $\tilde{O}(mn^{\sqrt{2}-1})$ Time}
In this section we are going to prove the modified version of Theorem \ref{thm:4approx}, where a $\log{M}$ factor is added to the running time with $M$ being the maximum edge weight.  

\begin{theorem}
\label{thm:modified_4approx}
{For every $\eps>0$, there is a $(4+\eps)$-approximation algorithm for the girth in directed graphs with edge weights in $\{1,\ldots,M\}$ that runs in $\tilde{O}(mn^{\sqrt{2}-1}\log(M) /\eps)$ time.}
\end{theorem}
\begin{proof}
Suppose that we want an $\tilde{O}(mn^\alpha)$ time girth approximation algorithm. 
Let $\beta = 2(1+\epsilon)$. 
 As a first step, we sample a set $Q$ of $\tilde{O}(n^{\alpha})$ vertices and do in and out Dijkstra from them.

We let $V'_i=\{v\in V~|~\exists q\in Q:~d(v,q)\leq \beta(1+\eps)^{i+1} \textrm{ and } d(q,v)\leq \beta(1+\eps)^{i+1}\}$. If $V_i'\neq \emptyset$ for some $i$, then we have that the girth $g$ is at most $2\beta(1+\epsilon)^{i+1}$. If $(1+\epsilon)^i\le g\le (1+\epsilon)^{i+1}$, this is a $2\beta(1+\epsilon)\le 4(1+3\epsilon)=4+O(\epsilon)$ approximation.

So take the minimum $i$ where $V'_{i+1}\neq \emptyset$. Let $g'=(1+\epsilon)^{i+1}$ be our current upper bound for the girth $g$. We initially mark all vertices ``on", meaning that they are not processed yet. For each on vertex $u$, we either find the smallest cycle of length at most $g'$ passing through $u$ where all vertices of the cycle are on, or conclude that there is no cycle of length at most $g'$ passing through $u$. When a vertex $u$ is processed, we mark it as ``off". We proceed until all vertices are off.





We apply Lemma \ref{lemma:modifiedDijkstra} for $\beta = 2(1+\epsilon)$. Note that since $V'_i=\emptyset$, $Z^j_i(u)=B^j(u)$ is all the vertices at distance $[(1+\epsilon)^j, (1+\epsilon)^{j+1})$ from $u$. The lemma outputs sets $R_i^j(u)\subseteq Z^j_i(u)$, where $|R_i^j(u)|=O(\log^2{n})$ and the number of vertices in $B^j(u)$ at distance $\beta g'$ from $R_i^j(u)$ is at most $O(n^{1-\alpha})$ whp.  
Fix some on vertex $u$. We do modified Dijkstra from $u$ up to vertices with distance at most $g'/2$ from $u$ as follows:

We begin by placing $u$ in the Fibonacci heap with $d[u]=0$ and all other on vertices with $d[\cdot]=\infty$. When a vertex $x$ is extracted from the heap with estimate $d[x]$, we determine the $j$ for which $(1+\eps)^j\leq d[x]<(1+\eps)^{j+1}$; here $j$ could be the boundary case that we called $\emptyset$ if $d[x]=0$.
Then we check whether $d(x,r)\le g' - (1+\epsilon)^{j} + (1+\epsilon)^{j'+1}$ for all $r\in R^{j'}_i(u)$ for all $j'$. If $x$ does not satisfy this condition, we ignore it and extract a new vertex from the heap. Otherwise, we go through all its out-edges $(x,y)$, and if $d[y]>d[x]+w(x,y)$, we update $d[y]=d[x]+w(x,y)$. 
We stop when the vertex $u$ extracted from the heap has $d[u]> g'/2$. 

Let $S_i(u)$ be the set of all the vertices visited in the modified out-Dijkstra. Simillarly, let $T_i(u)$ be all the vertices visited in the analogous modified in-Dijkstra (using an analogous version of Lemma \ref{lemma:modifiedDijkstra}). 

Suppose that there is a vertex $v$ with $\rd{u}{v}\le g'$, where all vertices in the $uv$ cycle $C$ are on. Without loss of generality, suppose that $d_C(u,v)\le g'/2$. So $d(u,v)\le g'/2$. Moreover, suppose that $v\in Z_i^j(u)$, i.e. $(1+\epsilon)^j\le d(u,v) \le (1+\epsilon)^{j+1}$. So for any vertex $w\in Z_i^{j'}(u)$ for some $j'$ we have that $
  d(v,w)\le d(v,u)+d(u,w) \le g'-(1+\epsilon)^j+(1+\epsilon)^{j'+1}. 
$ Since all vertices on the $uv$ path that is part of the cycle are on and the length of this path is at most $g'/2$, we visit $v$ in the out-Dijkstra, i.e. $v\in S_i(u)$. Similarly, if $d(v,u)\le g'/2$, we visit $v$ in the in-Dijkstra and so $v\in T_i(u)$.


If both $S_i(u)$ and $T_i(u)$ have size at most $n^\alpha$, we do Dijkstra from $u$ in the induced subgraph on $S_i(u)\cup T_i(u)$, and see if there is a cycle of length at most $g'$ passing through $u$ (and find the smallest such cycle), which takes $O(\frac{m}{n}n^\alpha)$ time. We take the length of this cycle as one of our estimates. The modified in and out Dijkstras take $O(\log^2{n}.\frac{\log{nM}}{\eps}.n^\alpha.\frac{m}{n})$, as checking the conditions for each $x$ extracted from the heap takes $O(\log^2{n}.\frac{\log{nM}}{\eps})$ time.
So in $\tilde{O}(\frac{\log{M}}{\eps}n^\alpha.\frac{m}{n})$ time we process $u$ and mark it as "off and proceed to another vertex.




Suppose $S_i(u)$ has size bigger than $n^\alpha$ (the case where $T_i(u)$ has size bigger than $n^\alpha$ is similar). Note that by Lemma \ref{lemma:modifiedDijkstra} we have $|S_i(u)|\le O(n^{1-\alpha})$ because for each $r\in R_i^j(u)$, we have that $d(x,r)\le g'-(1+\eps)^j+(1+\eps)^{j+1}\le g'+\eps(1+\eps)^j\le g'+\eps g'/2\le \beta g'$. So it is a subset of vertices that are at distance at most $\beta g'$ from all samples in $R_i^j$ for all $j$. Our new goal is the following: 
 We want to either find the smallest cycle of length at most $g'$ passing through $S_i(u)$ that contains no off vertices, or say that there is no cycle of length $\le g'$ passing through any of the vertices in $S_i(u)$ whp. 

For this, we do another Modified Dijkstra from $u$ as follows:

We begin by placing $u$ in the Fibonacci heap with $d[u]=0$ and all other on vertices with $d[\cdot]=\infty$. When a vertex $x$ is extracted from the heap with estimate $d[x]$, we determine the $j$ for which $(1+\eps)^j\leq d[x]<(1+\eps)^{j+1}$; here $j$ could be the boundary case that we called $\emptyset$ if $d[x]=0$.
Then we check whether $d(x,r)\le \beta g'=2(1+\eps)g'$ for all $r\in R_i^j(u)$. If it is not, we ignore it and extract a new vertex from the heap. Otherwise, we go through all its out-edges $(x,y)$, and if $d[y]>d[x]+w(x,y)$, we update $d[y]=d[x]+w(x,y)$. We stop when the vertex $u$ extracted from the heap has $d[u]> 3g'/2$.

We show that if there is a cycle of length at most $g'$ going through $v\in S_i(u)$ containing to off vertex, all vertices of the cycle are among the vertices we visit in the modified Dijkstra: Suppose that $\rd{w}{v}\le g'$, and suppose that $v\in Z_i^j(u)$ and $w\in Z_i^{j'}(u)$. Then for every $r\in R_i^{j'}(u)$, we have that $d(w,r)\le d(w,v)+d(v,r)\le g'-d(v,w)+ g'-(1+\eps)^j+(1+\eps)^{j'+1}$. Since $d(v,w)\ge (1+\eps)^{j'}-(1+\eps)^{j+1}$, we have $d(w,r)\le 2g'+\eps(1+\eps)^{j'}+\eps(1+\eps)^j\le 2g'+3\eps g'/2+\eps g'/2 =\beta g'$. Since the $uw$ path that goes through $v$ is a path of length at most $\beta g'$ that has no off vertices, we visit $w$ in the modified Dijkstra. 

By Lemma \ref{lemma:modifiedDijkstra} the total number of vertices visited in the modified Dijkstra is at most $O(n^{1-\alpha})$. Let the subgraph on these vertices be $G'$. We recurse on $G'$, and find a $4+O(\eps)$ approximation of the girth in $G'$. The girth in $G'$ is a lower bound on the minimum cycle of length $\le g'$ passing through any vertex in $S_i(u)$ that has no off vertex. We take this value as one of our estimates. So we have processed all vertices in $S_i(u)$ and we mark them off. This takes $O(\frac{m}{n}.\frac{\log{M}}{\eps}.((n^{1-\alpha})^{1+\alpha}))$, and we have marked off at least $n^\alpha$ vertices. So we spend $O(\frac{m}{n}.\frac{\log{M}}{\eps}.n^{1-\alpha^2-\alpha})$ for processing each vertex. 
Letting $1-\alpha^2-\alpha=\alpha$, we have that $\alpha = \sqrt{2}-1$. So the total running time is $\tilde{O}(m n^{\sqrt 2 -1} \log(M)/\eps)$. Our final estimate of the girth is the minimum of all the estimates we get through processing vertices.
\end{proof}


\subsection{$(2k+\epsilon)$-Approximation Algorithm For the Girth}

In this section we are going to prove a modified version of Theorem \ref{thm:genk}, where a $\log{M}$ factor is added to the running time with $M$ being the maximum edge weight. The proof is a generalization of the proof of Theorem \ref{thm:modified_4approx}. 

\begin{theorem}\label{thm:modified_genk}
For every $\eps>0$ and integer $k\geq 1$, there is a $(2k+\eps)$-approximation algorithm for the girth in directed graphs with edge weights in $\{1,\ldots,M\}$ that runs in $\tilde{O}(mn^{\alpha_k}\log(M) /\eps)$ time, where $\alpha_k>0$ is the solution to $\alpha_k(1+\alpha_k)^{k-1}=1-\alpha_k$.
\end{theorem}

Suppose that we are aiming for a $2k(1+O(\epsilon))$ approximation algorithm for the girth, in $\tilde{O}(mn^{\alpha}\log{M}/\eps)$ time, where we set $\alpha$ later. So basically we want to spend $\tilde{O}(\frac{m}{n}\frac{\log{M}}{\eps}n^{\alpha})$ per vertex. Let $\beta = k+k^2\epsilon+k\epsilon=k+O(\eps)$.
As before, first we sample a set $Q$ of $\tilde{O}(n^{\alpha})$ and do in and out Dijkstra from each vertex $q\in Q$. Let $i_{min}$ be the minimum number $i$ such that the set $V'_i=\{v\in V~|~\exists q\in Q:~d(v,q)\leq \beta(1+\eps)^{i+1} \textrm{ and } d(q,v)\leq \beta(1+\eps)^{i+}\}$ is non-empty. So our initial estimate of the girth is $2\beta(1+\epsilon)^{i_{min}+1}$. 

Let $i=i_{min}-1$ and let $g'=(1+\eps)^{i+1}$ be our estimate of the girth. Initially we mark all vertices as ``on", and as we process each vertex, we either find a smallest cycle of length at most $g'$ with no ``off" vertex, or we say that there is no cycle of length at most $g'$ passing through it whp, and we mark the vertex as off.

We apply Lemma \ref{lemma:modifiedDijkstra} for $\beta = k+k^2\epsilon+k\epsilon$ and the set $Q$ as input. It gives us the sets $R_i^j(u)$ of size $O(\log^2{n})$ for all $j$, such that the number of vertices in $\bar{B}^j(u)=\{w\in V|d(u,w)\le (1+\epsilon)^{j+1}\}$ that are at distance at most $\beta g'$ from all $r\in R_{i}^j(u)$ is at most $O(n^{1-\alpha})$ whp. 


We take an on vertex $u$ and do ``modified" Dijkstra from (to) $u$, stopping at distance $g'/2$, such that the set of vertices we visit contains any cycle of length $g'$ that passes through $u$ that has no off vertex. We explain this modified Dijkstra later.

We call the set of vertices that we visit in the modified out-Dijkstra $S_i^1(u)$. If $S_i^1(u)\le n^\alpha$, we do an analogous modified in-Dijkstra from $u$, and let $T_i^1(u)$ be the set of vertices visited in this in-Dijkstra. If $T_i^1(u)\le n^\alpha$, then we do Dijkstra from $u$ in the subgraph induced by $S_i^1(u)\cup T_i^1(u)$, and hence find a smallest cycle of length $\le g'$ that passes through $u$ with no off vertex. We take the length of this cycle as one of our estimates for the girth. If there is no such cycle, we don't have any estimate from $u$. Now we mark $u$ as off and proceed the algorithm by taking another on vertex. Our modified Dijkstras takes $O(\log^2{n}.\frac{\log{Mn}}{\eps}.\frac{m}{n}|S|)$ time if $S$ is the set of vertices visited by the Dijkstra.
Hence for processing $u$ we spend $O(\log^2{n}.\frac{\log{Mn}}{\eps}.\frac{m}{n}n^\alpha)$ time. 

So suppose that either $S_i^1(u)$ or $T_i^1(u)$ have size bigger than $n^\alpha$. Without loss of generality assume that $|S_i^1(u)|\ge n^\alpha$ (the other case is analogous). For $1\le l\le k$, define sets $S_i^l(u)$ as the set of on vertices $w\in V$ such that there is a path of length at most  $(2l-1)g'/2$ from $u$ to $w$ that contains no off vertex, and if $w\in B^j(u)$, then for all $r\in R_i^{j'}(u)$ for all $j'$, we have $d(w,r)\le (l+l^2\eps)g' + (1+\epsilon)^{j'+1}-(1+\epsilon)^{j}$. Once we explain our modified Dijkstras, it will be clear that $S_i^1$ defined here is indeed the set of vertices visited in the first modified out-Dijkstra.

We set $S_i^0(u)=\{u\}$. We prove the following useful lemma in the Appendix.

\begin{lemma}
\label{lemma:klevelsets}
For all $l\in \{1,\ldots,k\}$, we have that $S_i^{l-1}(u)\subseteq S_i^{l}(u)$. Moreover, if $w\in V$ is in a cycle of length at most $g'$ with some vertex in $S_i^{l-1}(u)$ such that the cycle contains no off vertex,
then we have $w\in S_i^{l}(u)$. 
\end{lemma}

Our algorithm will do at most $k$ modified Dijkstras from $u$, where we prove that the set of vertices visited in the $l$th Dijkstra is $S_i^l(u)$. After performing each Dijkstra we decide if we continue to the next modified Dijkstra from $u$ or proceed to another on vertex.

Suppose that at some point we know that the set $S_i^{l-1}(u)$ is the set of vertices visited in the $(l-1)$th modified Dijkstra, and we want to proceed to the $l$th Dijkstra. Our new goal is the following: We want to catch a minimum cycle of length $\le g'$ passing through $S_i^l$ with no off vertex. For this, we do the $l$th modified Dijkstra form $u$ as follows. 

We begin by placing $u$ in the Fibonacci heap with $d[u]=0$ and all other on vertices with $d[\cdot]=\infty$. When a vertex $x$ is extracted from the heap with estimate $d[x]$, we determine the $j$ for which $(1+\eps)^j\leq d[x]<(1+\eps)^{j+1}$; here $j$ could be the boundary case that we called $\emptyset$ if $d[x]=0$.
Then we check whether $d(x,r)\le (l+l^2\eps)g'- (1+\epsilon)^{j} + (1+\epsilon)^{j'+1}$ for all $r\in R^{j'}_i(u)$ for all $j'$. If $x$ does not satisfy this condition, we ignore it and extract a new vertex from the heap. Otherwise, we go through all its out-edges $(x,y)$, and if $d[y]>d[x]+w(x,y)$, we update $d[y]=d[x]+w(x,y)$. 
We stop when the vertex $u$ extracted from the heap has $d[u]> (2l-1)g'/2$. 

It is clear by definition that the set of vertices that this modified Dijkstra visits is $S_i^l(u)$.
Now if $|S_{i}^l(u)|\le c(|S_i^{l-1}(u)|.n^\alpha)^{\frac{1}{1 +\alpha}}$ for some constant $c$, we recurse on the subgraph induced by $S_{i}^l(u)$, i.e. $G[S_{i}^l(u)]$, to get an $2k+O(\eps)$ approximation of the girth on this subgraph. The girth in $G[S_{i}^l(u)]$ is a lower bound on the minimum cycle of length $\le g'$ passing through $S_{i}^{l-1}(u)$ with no off vertex. So we take this value as one of our estimates
and we mark all vertices of $S_i^{l-1}(u)$ as off. The running time of this recursion is $\tilde{O}(\frac{m}{n}\frac{\log{M}}{\eps}|S_i^l(u)|^{1+\alpha})$ as the average degree is $O(\frac{m}{n})$. Since we process $S_{i}^{l-1}(u)$ vertices in this running time, we spend $\tilde{O}(\frac{m}{n}\frac{\log{M}}{\eps}.|S_i^l(u)|/|S_i^{l-1}(u)|)\le \tilde{O}(\frac{m}{n}.\frac{\log{M}}{\eps}.n^\alpha)$ for each vertex.

Note that $|S_i^k(u)|\le O(n^{1-\alpha})$. This is because for all $x\in S_i^{k}\cap B^{j}(u)$ and for all $r\in R_i^j(u)$, we have that $d(x,r)\le (k+k^2\eps)g'+(1+\epsilon)^{j+1}-(1+\epsilon)^j\le (k+k^2\eps)g'+\epsilon(1+\epsilon)^{j}\le (k+k^2\eps)g'+\epsilon(2k-1)g'/2 \le (k+k^2\epsilon+k\eps)g'=\beta g'$. So $S_i^k(u)$ is a subset of all vertices in $B^{j}(u)$ with distance at most $\beta g'$ from all $r\in R_i^j(u)$, and so by Lemma \ref{lemma:modifiedDijkstra} it has size at most $O(n^{1-\alpha})$.

When all vertices are marked off, we take the minimum value of all the estimates as our estimate for $g$.

Since we have that $S_i^l(u)\le O(n^{1-\alpha})$, if we set $\alpha$ appropriately, for some $l<k$ we have that $|S_{i}^{l+1}(u)|\le (|S_i^l(u)|.n^\alpha)^{\frac{1}{1+\alpha}}$. For $k=1$, setting $\alpha=1/2$ gives us the algorithm of Theorem \ref{thm:2approx}. For $k>1$, the following lemma determines $\alpha$. The proof of the lemma can be found in the Appendix.

\begin{lemma}
\label{lemma:alphaformula}
For $k>1$, let the sets $S_i^{l}$ for $l=1,\ldots,k$ be such that $S_i^l\subseteq S_{i}^{l+1}$ for all $l<k$, $S_i^1\ge n^\alpha$ and $S_i^k\le O(n^{1-\alpha})$. Let $0<\alpha<1$ satisfy 
$\alpha(1+\alpha)^{k-1}=1-\alpha$. Then there is $l<k$ and a constant $c$ such that $|S_i^{l+1}|\le c(|S_i^l|.n^\alpha)^{\frac{1}{1+\alpha}}$.
\end{lemma}

Note that for $k=2$, Lemma \ref{lemma:alphaformula} sets $\alpha=\sqrt{2}-1$ and thus gives us the algorithm of Theorem \ref{thm:modified_4approx}. 

 \subsection{Removing the $\log{M}$ factor}
 In this subsection we show how to remove the $\log{M}$ factor in the running times of our algorithms where $M$ is the maximum edge weight, resulting in strongly polynomial algorithms.

Assume that we have a $(2k+\epsilon)$-approximation algorithm $A$ for the girth in $\tO(mn^{\alpha_k}\log{M}/\epsilon)$ running time for some $0\le \alpha_k\le 1$. We want to obtain an algorithm that gives us a $(2k+O(\epsilon))$-approximation of the girth in $\tO(mn^{\alpha_k}/\epsilon)$ time. 

First, suppose that we know the smallest number $W$ such that there is a cycle with all edge weights at most $W$. Then by the definition of $W$ we have that $W\le g$ and $g\le nW$. Moreover, note that the edges of any cycle with total weight at most $(2k+O(\epsilon))g$ cannot have weights more than $3knW$, so we can remove any edge with weight more than $3knW$. Let $R=W\epsilon/n$. Let $H$ be a copy of $G$, with the weight $w_G(e)$ of the edge $e$ replaced by $w_H(e)=\floor{w_G(e)/R}$. Note that the weights of $H$ are bounded by $O(n^2/\epsilon')$. 

Now consider a cycle $C$ in $G$. Suppose that $C$ has $n_C$ edges. Let $w_G(C)$ and $w_H(C)$ be the sum of the edge-weights of $C$ in $G$ and $H$ respectively. For any edge $e$, we have that $w_G(e)-R\le R\cdot w_H(e)\le w_G(e)$. This gives us 
\begin{equation}
\label{eq:cyclesrelations}
w_G(C)-Rn_C\le R\cdot w_H(C)\le w_G(C).
\end{equation}

Note that if $C$ is the cycle with minimum length in $G$, then we have that $Rg'\le Rw_H(C)\le g$, where $g'$ is the girth of $H$. 

Now we apply algorithm $A$ on $H$, which takes $\tO(mn^{\alpha_k}/\epsilon)$ time. Suppose that it outputs a cycle $C$ such that $g'\le w_H(C)\le (2k+\eps)g'$. Since $g\ge Rg'$ and by equation \ref{eq:cyclesrelations} we have $w_G(C)\le Rw_H(C)+Rn_c\le (2k+\eps)Rg'+Rn\le (2k+2\eps)g$. The last inequality uses the fact that $Rn=W\eps \le g\eps$. 

It suffices to show how we obtain $W$. We sort the edges of $G$ in $\tO(m)$ time, so that the edge weight are $w_1\le \ldots\le w_m$. We find $W$ using binary search and DFS as follows: Suppose that we are searching for $W$ in the interval $w_i\le \ldots\le w_j$ for $1\le i\le j\le m$. Let $r=(i+j)/2$, we remove all the edges with weight more than $w_{r}$ and then do DFS in the remaining graph to see if it has a cycle. If it does, we update $j = r$, otherwise we update $i= r$. Note that this process takes $\tO(m)$ time.

\section{Hardness}
\label{sec:hardness}
Our hardness result is based on the following $k$-Cycle hypothesis (see \cite{lincolnsoda,ancona2019,maxdynamic}).

\begin{hypothesis}[$k$-Cycle Hypothesis] \label{hyp:cycle}
In the word-RAM model with $O(\log m)$ bit words,
for any constant $\eps > 0$, there exists a constant integer $k$, so that there is no $O(m^{2-\eps})$ time algorithm that can detect a $k$-cycle in an $m$-edge graph.
\end{hypothesis}

All known algorithms for detecting $k$-cycles in directed graphs with $m$ edges run at best in time $m^{2-c/k}$ for various small constants $c$ \cite{YuZw04,AlYuZw97,lincolnsoda,patternscycles19}, even using powerful tools such as fast matrix multiplication. Refuting the $k$-Cycle Hypothesis above would resolve a big open problem in graph algorithms. Moreover, as shown by Lincoln et al. \cite{lincolnsoda} any algorithm for directed 
$k$-cycle detection, for $k$-odd,
with running time $O(mn^{1-\eps})$ for $\eps>0$ whenever $m=\Theta(n^{1+2/(k-1)})$ would imply an $O(n^{k-\delta})$ time algorithm for $k$-clique detection for $\delta>0$. If the cycle algorithm is ``combinatorial'', then the clique algorithm would be ``combinatorial'' as well, and since all known $O(n^{k-\delta})$ time $k$-clique algorithms use fast matrix multiplication, such a result for $k$-cycle would be substantial.

We will show that under Hypothesis~\ref{hyp:cycle}, approximating the girth to a factor better than $2$ would require $mn^{1-o(1)}$ time, and so up to this hypothesis, our approximation algorithm is optimal for the girth in unweighted graphs.

\begin{theorem}
Suppose that for some constants $\eps>0$ and $\delta>0$, there is an $O(m^{2-\eps})$ time algorithm that can compute a $(2-\delta)$-approximation of the girth in an $m$-edge directed graph. Then for every constant $k$, one can detect whether an $m$-edge directed graph contains a $k$-cycle, in $O(m^{2-\eps})$ time, and hence the $k$-Cycle Hypothesis is false.
\end{theorem}

\begin{proof}
The proof is relatively simple.
Suppose that for some constants $\eps>0$ and $\delta>0$, there is an $O(m^{2-\eps})$ time algorithm that can compute a $(2-\delta)$-approximation of the girth in an $m$-edge directed graph.

Now let $k\geq 3$ be any constant integer and let $G$ be an $n$-node, $m$-edge graph. First randomly color each vertex of $G$ with one of $k$ colors. Let $C$ be any $k$-cycle in $G$. With probability $1/k^k$, for each $i=0,\ldots, k-1$, the $i$th vertex of $C$ is colored $i$. 

Now, for each $0\leq i\leq k-1$, let $V_i$ be the vertices colored $i$. For each vertex $u\in V_i$, and each directed edge $(u,v)$ out of $u$, keep $(u,v)$ if and only if $v\in V_{i+1}$ where the indices are taken mod $k$. This builds a graph $G'$ which is a subgraph of $G$ and contains a $k$-cycle if $G$ does with probability $\geq 1/k^k$.

$G'$ has two useful properties. (1) Any cycle of $G'$ has length divisible by $k$, and (2) (which follows from (1)) the girth of $G'$ is $k$ if $G'$ contains a $k$-cycle and it is $\geq 2k$ otherwise.

As $G'$ has at most $m$ edges (it is a subgraph of $G$), we can use our supposedly fast $2-\delta$ approximation algorithm to determine whether the girth is $k$ or larger in $O(m^{2-\eps})$ time.
By iterating the construction $O(k^k\log n)$ times, we get that the $k$-cycle problem in $G$ can be solved in $\tilde{O}(k^k m^{2-\eps})$ time, and as $k$ is a constant, we are done. The approach can be derandomized with standard techniques (e.g. \cite{colorcodingj}).
\end{proof}
\section{Acknowledgements}
We thank the anonymous reviewers for their insightful comments.
\bibliographystyle{plain}
\bibliography{references}

\section{Appendix}
\subsection{Omitted proofs}
\label{subsec:omittedproofs}

\begin{proof}[Proof of Lemma~\ref{lemma:unwtd}]

We start with a simple claim which is proved at the end:

\begin{claim}
Let $q\geq 2$ be an integer. Let $L\geq 1$ be an integer.
There is a directed rooted tree with $\leq 3L$ nodes, $L$ leaves, with every node of outdegree $\leq q$ and such that every root to leaf path has the same length $\lceil \log_q L \rceil$.\label{claim:tree}
\end{claim}

The idea of the proof is to represent every edge $(u,v)$ of $G$ by a $t$-length path from $u$ to $v$ via some auxiliary nodes, so that the total number of auxiliary nodes is small, and the degree of every node is small as well.

Let $q=\max\{2,\lceil m/n\rceil\}$. 
Consider some node $u$ and its out-neighbors $v_1,\ldots,v_{deg(u)}$.
Remove the edge from $u$ to $v_j$ for each $j$. Let $d$ be the smallest power of $q$ that is larger than $deg(u)$, i.e. $q^{d-1}< deg(u)\leq q^d$ and $d=\lceil \log_q deg(u) \rceil$.

Using the construction of Claim~\ref{claim:tree}, create a partial $q$-ary tree $T_u$ of at most $3\lceil deg(u)/q \rceil$ auxiliary nodes, with $\lceil deg(u)/q \rceil$ leaves,  and so that the leaves are all at depth $\lceil\log_q (\lceil deg(u)/q\rceil)\rceil= d-1$. Then, make the original out-neighbors $v_1,\ldots,v_{deg(u)}$ of $u$ children of the leaves of $T_u$ so that every leaf of $T_u$ has at most $q$ children.

Let $t=\lceil \log_q n \rceil$.
Notice that since $deg(u)<n$, we have that $d\leq t$. If $d=t$, set $r_u=u$. If $d<t$, add another $t-d$ new auxiliary nodes $u_1,\ldots,u_{t-d}$, connect them into a directed path $u_1\rightarrow \ldots \rightarrow u_{t-d}$ and then add the edge $(u_{t-d},r_u)$. Let $u_1=u$. This completes a directed tree $T(u)$ rooted at $u$ such that the number of edges on any root-to-leaf path is $t$.
See Figure \ref{fig:trees} for example trees.

\begin{figure}[h]
  \centering
  \includegraphics[width=.7\linewidth]{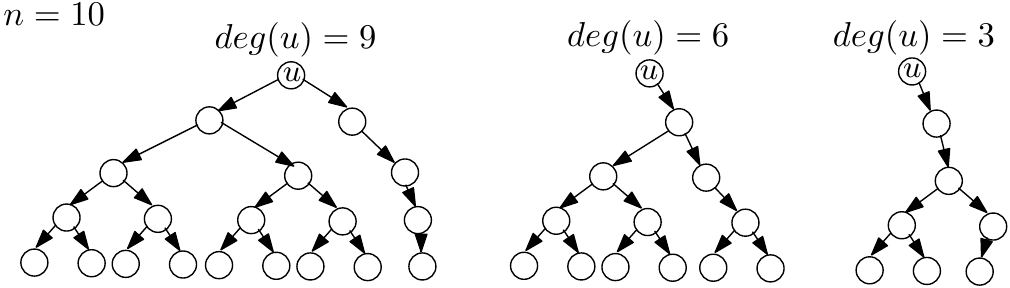}
  \caption{Here we give examples of the construction of $T(u)$ when $t=4$ (e.g. when $n=10$), and when the out-degree of $u$ is $9,6$ and $3$.}
  \label{fig:trees}
\end{figure}

The obtained graph is unweighted. Notice that for each $j$, the original edge $(u,v_j)$ is replaced by a path in $T(u)$ of length exactly $t$, and hence for every $u,v\in V$, $d_{G'}(u,v)=td_G(u,v)$. 

Since the auxiliary nodes do not create new cycles, any cycle $C$ in $G'$ must correspond to a cycle in $G$ that can be obtained from $C$ by replacing each subpath between nodes of $V$ with the edge corresponding to it, and the girth of $G'$ is exactly $t$ times the girth of $G$.
Similarly, if we had a $c$-roundtrip spanner over the new graph $G'$, we can obtain a $c$-spanner of $G$ by replacing each auxiliary path between vertices of $V$ with the corresponding edge of $G$. The number of edges does not increase.

Every vertex in the new graph has out-degree at most $q$. If we would like the in-degrees to be bounded by $q$ as well, we can perform the same procedure (with edge directions reversed) on the in-neighborhoods.

The total number of auxiliary vertices added to $T(u)$ is $$t+3\lceil deg(u)/q\rceil \leq \log n + 3+3n\cdot deg(u)/m.$$
Over all vertices the total number of auxiliary vertices is at most $$\sum_{u\in V} \left[\log n + 3+3\frac{n\cdot deg(u)}{m}\right]=O(n\log n).$$
\end{proof}

\begin{proofof}{Claim \ref{claim:tree}}
It is easy to see that if $q^{d-1}<L\leq q^d$, we can always take a complete $q$-ary tree on $q^d$ leaves and remove enough leaves until we only have $L$. This would definitely achieve the depth requirement. However, if we are not careful, we might have more than $3L$ nodes in the tree. Here we do a more fine-tuned analysis to have both the size and the depth of the tree under control.

Let us consider the $q$-ary representation of $L$:
$L=a_{d-1}q^{d-1}+a_{d-2}q^{d-2}+\ldots+a_0.$
Here each $a_j\in \{0,\ldots,q-1\}$.

We will show inductively how to build a rooted tree with out-degree $\leq q$ so that every leaf is at depth $d$.
The base case is when $d=1$, so that $L=a_0$. Then we simply have a root with $a_0$ children.

Suppose that $d>1$. Let us assume that for every integer $\ell<q^{d-1}$ we can create a rooted tree with outdegree at most $q$, $\ell$ children all of depth $d-1$.
Consider now $L=a_{d-1}q^{d-1}+a_{d-2}q^{d-2}+\ldots+a_0.$
Create a root $r$ with $a_{d-1}+1$ children. The first $a_{d-1}$ children are roots of complete $q$-ary trees with $q^{d-1}$ leaves. These have depth $d-1$, and together with the edge from $r$ to their roots, they have depth $d$. The last child of $r$ is a root of a directed tree formed inductively to have $a_{d-2}q^{d-2}+\ldots+a_0$ leaves (all of depth $d-1$) and out-degree $\leq q$.
As $d=\lceil \log_q L\rceil$, we are done with the depth argument.

As for the number of nodes in the tree, we prove it by induction. The base case is when $d=1$, so $L=a_0$ and the number of nodes in the tree is $L+1\leq 3L$ (as $L\geq 1$). Suppose the number of nodes in the tree is $\leq 3L$ for all $L<q^{d-1}$. Consider $L=a_{d-1}q^{d-1}+a_{d-2}q^{d-2}+\ldots+a_0.$
The number of nodes in the tree is then at most
\[1+a_{d-1}\frac{q^d}{q-1}+3\cdot (L-a_{d-1}q^{d-1}),\]
where $1$ is for the root, $\frac{q^d}{q-1}$ is the number of nodes of a complete $q$-ary tree with $q^{d-1}$ leaves and $(L-a_{d-1}q^{d-1})$ is the number of leaves left after the first $a_{d-1}q^{d-1}$ are covered by the complete $q$-ary trees.
The expression above is
\[\leq 3L + 1+\frac{a_{d-1}}{q-1}(q^d-1-3q^{d}+3q^{d-1})\leq 3L + \frac{a_{d-1}}{q-1}(-2q^{d}+3q^{d-1}+q-2).\]
Now, since $d\geq 2$ ($d=1$ was the base case), and $q\geq 2$, we have that
$$-2q^{d}+3q^{d-1}+q-2=q^{d-1}(3-2q)+q-2\leq -q^{d-1}+q-2\leq -2<0,$$
and hence the number of nodes is $\leq 3L$.
\end{proofof}

\begin{proofof}{Lemma \ref{lemma:modifiedDijkstra}}
First suppose that we are able to pick a random sample $R_{i}^j(u)$ of $c\log n$ vertices from $\bar{Z}_i^j(u)$. Then we can define $B_{i}^j(u)=\{z\in \bar{Z}_i^j(u)~|~d(z,y)\leq d ,~\forall y\in R_i^j(u)\}$.

Consider any $s\in V$ with at least $0.2 |\bar{Z}_i^j(u)|$ nodes $v\in V$ so that $d(s,v),d(v,s)\leq d$. As $\bar{Z}_i^j(u)\geq 10n^{1-\alpha}$ (as otherwise we would be done and the sampled vertices would work), $0.2 |\bar{Z}_i^j(u)|\geq 2n^{1-\alpha}$, and so with high probability, for $s$ with the property above, $Q$ contains some $q$ with $d(s,q),d(q,s)\leq d$, and so $s\in V'_i$. Thus with high probability, for every $s\in Z_i$, there are at most $0.2 |\bar{Z}_i^j(u)|$
 nodes $v\in V$ so that $d(s,v),d(v,s)\leq d$.
 
We will iterate this sampling process until we arrive at a subset of $\bar{Z}_i^j(u)$ that is smaller than $10n^{1-\alpha}$ that contains all the vertices in $Z_i^j(u)$ with distance at most $d$ to all the sampled vertices, as follows: 
 
Let $Z^j_{i,0}(u)=\bar{Z}_i^j(u)$. For each $k=0,\ldots, 2\log n$, let $R^j_{i,k}(u)$ be a random sample of $O(\log n)$ vertices of $Z^j_{i,k}(u)$. Define $Z^j_{i,k+1}(u)=\{z\in \bar{Z}_i^j(u)~|~d(z,y)\leq d \forall y\in \cup_{\ell=0}^k R^j_{i,\ell}(u)\}$. 
We get that for each $k$, $|Z^j_{i,k}(u)|\leq 0.8^k |\bar{Z}_i^j(u)|$ so that at the end of the last iteration, $|Z^j_{i,2\log n}|\leq 10 n^{1-\alpha}$. Hence we get the set $Z'^j_i(u)$ that we are after as $Z^j_{i,2\log n}$.

It is not immediately clear how to obtain the random sample $R^j_{i,k}(u)$ from $Z^j_{i,k}(u)$ as $Z^j_{i,k}(u)$ is unknown. We do it in the following way. For each $i,j,k$ we independently obtain a random sample $S_{i,j,k}$ of $Z_i$ by sampling each vertex independently with probability $p=100\log n / n^{\alpha}$.
For each of the (in expectation) $O(n^\alpha \log^4(n))$ vertices in the sets $S_{i,j,k}$ we run Dijkstra's to and from them, to obtain all their distances.
 
Now, for a fixed $i$, $j\leq i$, $k$, to obtain the random sample $R^j_{i,k}(u)$ of the unknown $Z^j_{i,k}(u)$, we assume that we already have $R^j_{i,\ell}(u)$ for $\ell<k$, and define
\[T^j_{i,k}(u)=\{s\in S_{i,j,k}~|~s\in \bar{Z}_i^j(u) \textrm{ and } d(s,y)\leq d \forall y\in \cup_{\ell<k} R^j_{i,\ell}(u).\}\]
Forming the set $T^j_{i,k}(u)$ is easy since we have the distances $d(s,v)$ for all $s\in S_{i,j,k}$ and $v\in V$, so we can check whether $s\in \bar{B}^j(u)$ and $s\in Z_i$ (thus checking that $s\in \bar{Z}_i^j(u)$) and $d(s,y)\leq d \forall y\in \cup_{\ell<k} R^j_{d,\ell}(u)$ in polylogarithmic time for each $s\in S_{i,j,k}$.

Now since $S_{i,j,k}$ is independent from all our other random choices, $T^j_{i,k}(u)$ is  a random sample of $Z^j_{i,k}(u)$ essentially created by selecting each vertex with probability $p$. If $Z^j_{i,k}(u)\geq 100n^{1-\alpha}$, with high probability, $T^{j}_{i,k}(u)$ has at least $10\log n$ vertices so we can pick $R^j_{i,k}(u)$ to be a random sample of $10\log n$ vertices of $T^j_{i,k}(u)$, and they will also be a random sample of $10\log n$ vertices of $Z^j_{i,k}(u)$. So we let $R_{i}^j(u)=\cup_{k} R_{i,k}^j(u)$, which has size $O(\log^2{n})$. The running time of this sampling procedure comes from the Dijkstras we perform from $S_{i,j,k}$s and hence it is $\tilde{O}(mn^{\alpha})$.
\end{proofof}

\begin{proofof}{Lemma \ref{lemma:klevelsets}}
First it is clear that for all $v\in S_{i}^{l-1}(u)$ and each $r\in R_i^{j'}(u)$, we have $d(v,r)\le ((l-1)+\eps(l-1)^2)g'+(1+\epsilon)^{j'+1}-(1+\epsilon)^j\le (l+\eps l^2)g'+(1+\epsilon)^{j'+1}-(1+\epsilon)^j$, and so $v\in S_{i}^l(u)$. 

Now suppose that for $v\in S_i^{l-1}(u)$ and $w\in V$, we have $\rd{v}{w}\le g'$. Suppose that $v\in B^{j_1}(u)$, $w\in B^{j_2}(u)$. For $r\in R_i^{j_3}$, we have 
$$
d(w,r)\le d(w,v)+d(v,r)\le d(w,v)+((l-1)+(l-1)^2\eps)g'+(1+\epsilon)^{j_3+1}-(1+\epsilon)^{j_1}.
$$ 
If $j_1 \ge j_2$, then since $d(w,v)\le g'$, we have 
$$
d(w,r)\le (l+(l-1)^2\eps)g'+(1+\epsilon)^{j_3+1}-(1+\epsilon)^{j_1}\le (l+l^2\eps)g'+(1+\epsilon)^{j_3+1}-(1+\epsilon)^{j_2}.
$$
If $j_1< j_2$, then we have $d(w,v)\le g'-d(v,w)\le g'-[(1+\epsilon)^{j_2}-(1+\epsilon)^{j_1+1}]$. Using the fact that $(1+\eps)^{j1}\le (l-1)g'$ we have that 
\begin{align*}
d(w,r)&\le (l+(l-1)^2\eps)g'+(1+\epsilon)^{j_3+1}-(1+\epsilon)^{j_2}+\eps(1+\eps)^{j_1}\\
&\le (l+(l-1)^2\eps)g'+(1+\epsilon)^{j_3+1}-(1+\epsilon)^{j_2}+\eps (l-1)g'\\
&\le (l+l^2\eps)g'+(1+\epsilon)^{j_3+1}-(1+\epsilon)^{j_2}.
\end{align*}

We have that $d(u,w)\le d(u,v)+d(v,w)\le (2l-3)g'/2+g'\le (2l-1)g'/2$. If the path $vw$ contains no off vertex, then there $uw$ path passing through $v$ contains no off vertex and so there is a path of length at most $(2l-1)g'/2$ with all vertices. So $w\in S_i^l(u)$.
\end{proofof}

\begin{proofof}{Lemma \ref{lemma:alphaformula}}
Assume that $|S_i^k|\le Cn^{1-\alpha}$ for some constant $C>1$. Suppose that for all $l<k$, we have that $|S_i^{l+1}|>C(|S_i^l|n^\alpha)^{\frac{1}{1+\alpha}}$. Using $|S_i^1|> n^\alpha$, we have that $|S_i^{k}|>Cn^{\frac{\alpha}{(1+\alpha)^{1-k}}+\alpha\sum_{j=1}^{k-1}\frac{1}{(1+\alpha)^j}}$. Since $1-\frac{1}{(1+\alpha)^{k-1}}=\alpha\sum_{j=1}^{k-1}\frac{1}{(1+\alpha)^j}$ and we have that $\alpha(1+\alpha)^{k-1}=1-\alpha$ iff $\alpha+\frac{\alpha-1}{(1+\alpha)^{k-1}}=0$, we obtain that $|S_i^k|>Cn^{1-\alpha}$, which is a contradiction.
\end{proofof}

\end{document}